\documentclass[a4paper,10pt]{article}

\usepackage[ansinew]{inputenc}
\usepackage{amssymb, amsthm}
\usepackage{latexsym}
\usepackage{amsmath}
\usepackage{hyperref}
\usepackage{authblk}
\usepackage{float}

\newtheorem{theo}{Theorem}[section]
\newtheorem{prop}[theo]{Proposition}

\newtheorem{lem}[theo]{Lemma}
\theoremstyle{definition}

\newcommand{\sfrac}[2]{{\textstyle\frac{#1}{#2}}}

\makeatletter
\renewcommand{\@maketitle}{
\newpage
 \null
 \vskip 1em%
 \begin{flushright}
ITP-UH-17/10\\
\end{flushright}
 \begin{center}%
  {\Large \textbf \@title \par}%
  \vskip 2.5em
  {\large \@author \par}
  \vskip 1.5 em
 \end{center}%
 \par} \makeatother

\title{Homogeneous heterotic supergravity solutions with linear dilaton}
\author{Christoph N\"olle}
\date{}

\begin{document}
\maketitle

\begin{center}
{\em Institut f\"ur Theoretische Physik, Leibniz Universit\"at Hannover \\
    Appelstra{\ss}e 2, 30167 Hannover, Germany } \\
   email: noelle@math.uni-hannover.de
\end{center}
\bigskip

\begin{abstract}
 I construct solutions to the heterotic supergravity BPS-equations on products of Minkowski space with a non-symmetric
coset. All of the bosonic fields are homogeneous and non-vanishing, the dilaton being a linear function on the
non-compact part of spacetime.
\end{abstract}

\bigskip

\tableofcontents
\numberwithin{equation}{section}

\section{Introduction}
 The aim of this paper is to present some homogeneous solutions to the heterotic supergravity equations \cite{BghdR89, BeckerSethi09}
 \begin{equation}\label{IntroSUSYEqtns}
  \begin{aligned}{}
    \nabla_\mu^-\varepsilon=\big(\nabla_\mu - \frac 18 H_{\mu\alpha\beta}\gamma^{\alpha\beta} \big)\varepsilon &=0  \\
    \gamma\big(d\phi - \frac 1{12} H\big)\varepsilon &=0  \\
    \gamma(F)\varepsilon &=0  
  \end{aligned}
 \end{equation} 
 and the Bianchi identity
 \begin{equation}\label{IntroBianchi}
   dH = \frac {\alpha'}4 \text{tr}\Big(R^+ \wedge R^+ - F\wedge F\Big),
 \end{equation} 
 where $\varepsilon$ is a spinor on a 10-dimensional Lorentzian manifold $M$, $H$ is a three-form, $\phi$ a function, 
and $F$ the curvature of a gauge field on $M$, and $\nabla^\pm$ are two connections on the tangent bundle $TM$,
involving $H$. In \eqref{IntroSUSYEqtns} $\gamma$ is the map from forms to the Clifford algebra. The manifold $M$ will be chosen of the form $M=
\mathbb R^{p,1}\times G/K$ for a non-symmetric, naturally reductive coset $G/K$, with (mostly) simple compact Lie groups
$K\subset G$, equipped with the metric induced by the Killing form.

\bigskip

This choice is canonical for the following reasons. First of all every such non-symmetric coset carries a
$G$-invariant 
(or homogeneous) three-form, which we will identify with $H$. Upon proper normalization of $H$, the question whether the gravitino equation $\nabla ^- \varepsilon=0$ has a solution $\varepsilon$, turns into a simple representation-theoretical
problem. Furthermore, upon this choice of $H$ we also get a solution for $F$, namely the curvature $R^-$ of the so-called canonical connection $\nabla^-$, which also appears in the gravitino equation. It satisfies both $\gamma(R^-)\varepsilon=0$ if $\varepsilon$ solves the gravitino equation, and $dH\sim $ tr$(R^+\wedge R^+ - R^-\wedge R^-)$, leading to a solution of all of the equations except the dilatino one, $\gamma(d\phi-\frac 1{12}H)\varepsilon =0$. 

\bigskip

 All this suggests that the heterotic supergravity equations are tailored to admit homogeneous solutions; 
in particular the Bianchi identity allowing for non-trivial $dH$ is an important deviation from the standard
supergravity rule $dH=0$, which would immediately rule out these spaces. The situation changes with the dilatino
equation however. On a general coset $G/K$ there are no homogeneous 1-forms which could serve as $d\phi$, and if they
exist they tend to be non-exact. Therefore the only obvious choice would be to take $d\phi=0$. This is not possible
however, as we have $\gamma(H)\varepsilon \neq 0$, and there is also a simple no-go theorem excluding this type of
solutions. 

\bigskip
 
 In \cite{LNP10} we proposed to circumvent this problem by allowing for non-trivial fermion condensates, but in this paper a solution based on purely bosonic backgrounds will be presented. The backdoor we are going to use is to introduce an additional $\mathbb R^p$ factor to our spacetime $M$, with 
$p$ the rank difference of $G$ and $K$ (or $p=2$ for equal rank groups). Then we can take $\phi$ to
be a linear function on $\mathbb R^p$, giving rise to a constant and thus homogeneous 1-form $d\phi$. It will be shown
below that upon this choice of $\phi$ it is often possible to solve also the last equation. The amount of supersymmetry
preserved in the space orthogonal to $\mathbb R^p \times G/K$ is then at least $\mathcal N= 2^p -p$. For trivial $K$
one gets a Wess-Zumino-Witten (WZW) model coupled to a linear dilaton, and these models exist also in type II string
theory. They were considered already in \cite{Callan91}, as a certain limit of NS5-branes.

\bigskip

 The BPS-equations \eqref{IntroSUSYEqtns} actually guarantee that our supergravity vacua preserve some supersymmetry, and 
it was sometimes argued that they imply the usual equations of motions. Ivanov has proven that this is not the
case if the equations of motion are truncated at order $\alpha'$ as well, instead one would have to replace $R^+$ in the Bianchi
identity and the equation of motion by $R^-$ to ensure this \cite{Ivanov}. The correct interpretation of this result seems to be that
full compatibility between the supersymmetry equations \eqref{IntroSUSYEqtns} and the equations of motion requires the full tower of string corrections to both sets of equations, as explained in \cite{BeckerSethi09}, based on results of \cite{BghdR89}.
 Therefore, maybe one should take the solutions of the system \eqref{IntroSUSYEqtns} not as a proof but rather as an 
indication that there exists a heterotic string theory on these backgrounds. 

\bigskip

 From a physical perspective the linear dilaton certainly rules out these spaces as models for our universe; but see
\cite{Kimura09} for an intersecting brane scenario with chiral fermions. An interpretation of the type of solution
considered here in terms of a decoupling limit of string theory is given in \cite{Aharony98, Giv99}. \\
\indent On the other hand, a couple of homogeneous solutions to the above equations have been presented in recent
years where the dilaton is actually constant \cite{FIUV, Ugarte09}. The method used in these
works is somewhat different from ours, as they take as starting point Strominger's reformulation of the BPS equations
\cite{Strominger}. Furthermore they avoid the above-mentioned no-go
theorem by choosing $M$ to be a non-semisimple Lie group (or a finite quotient thereof)
equipped with a metric of negative scalar curvature, whereas in our models the metric always comes from the
bi-invariant one on $G$ and has positive scalar curvature. Although many of the spaces we discuss
allow for other homogeneous metrics as well, the solution of the equations becomes more involved with these. One
advantage of not relying on Strominger's equations is that we are not restricted to compact spaces of dimension six. In
fact we will find solutions with compact spaces of arbitrary odd dimension, and also 6-dimensional ones.

\bigskip

 After introducing the necessary tools for homogeneous spaces in section \ref{sec_HomVecBdls}, we will in section
\ref{sec_Spinors} develop a method which allows us
 both to prove existence of $\nabla^-$-parallel spinors on many of the considered spaces, and to calculate the action of
$\gamma(H)$ on these spinors, thus enabling us to determine the linear dilaton needed to satisfy also the dilatino
equation. In section \ref{sec_Sasaki} we discuss homogeneous Sasaki-Einstein manifolds, which are a particular class of
spaces where this method can be applied.
 The last section \ref{sec_bsp} has some examples treated in detail, based on the cosets
 \begin{itemize}
 \item SU($n+1)$/SU($n)=S^{2n+1}$
 \item Sp($n+1)$/Sp$(n)=S^{4n+1}$
 \item Sp($n$)/SU($n$)
 \item SO($2n)$/SU($n$)
 \item SO($n+1)$/SO($n-1$)
 \item Spin(7)/$G_2=S^7$ 
 \item $G_2$/SU(3)$\,=S^6$ 
 \item SU(3)/U(1)$\times $U(1)
 \item SO(5)/SO(3)$_\text{max}$
\end{itemize}
 It is intriguing that all of those spaces admit one of the following structures:  
\begin{itemize}
 \item nearly K\"ahler (in 6D),
 \item nearly parallel $G_2$ (in 7D),
 \item Sasaki-Einstein (in odd dimension),
 \item 3-Sasaki (in $4n+3$ dimensions),
\end{itemize}
 although the metric we use in most cases differs from the one defining this structure. 
 These are exactly the spaces whose cones admit parallel spinors, 
 and they play an important role in other types of string theory as well \cite{Acharya98, Boyer07}.
 The amount of supersymmetry preserved depends on the geometric type of the manifold, nearly K\"ahler, $G_2$, and
Sasaki-Einstein generically have $\mathcal N=1$, whereas 3-Sasakian spaces preserve more supersymmetry. \\
\indent It should be mentioned that
the method presented does not generalize to symmetric spaces, like $S^ n=$ SO($n+1)$/SO($n$) with its round metric; the
equation $\nabla \varepsilon =0$ does not have a solution there. 

\section{Heterotic supergravity}\label{sec_HetSUGRA}
 The low-energy limit of heterotic string theory is given by 10D $\mathcal N=1$ supergravity coupled to super
Yang-Mills. The bosonic part of the effective action is \cite{BeckerSethi09}
 \begin{equation}
   S = \int_M \Big(\text{Scal}^g  + 4|d\phi|^2 - \frac 12 |H|^2 + \frac {\alpha'}4 \text{tr}\big(|R^+|^2 - |F|^2\big)\Big) \text{Vol}^g,
 \end{equation}  
 where we adopt the widely-used convention to denote by tr a positive-definite form on a Lie algebra, 
in fact always minus the ordinary trace over tangent space in our examples. It leads to the following field equations
(to order $\alpha')$:
 \begin{equation}\label{EOMs}
  \begin{aligned}{}
    \text{Ric}_{\mu\nu} +2 (\nabla d\phi)_{\mu\nu} -\frac 14 H_{\mu\alpha\beta}{H_\nu}^{\alpha\beta}
 +\frac {\alpha'}4 \Big[  R^+_{\mu\alpha\beta\gamma} R_\nu^{+ \alpha\beta\gamma} - \text{tr}\big(F_{\mu\alpha} {F_\nu}^\alpha
  \big)\Big]&=0, \\
 \text{Scal} + 4 \Delta \phi  -4|d\phi|^2 -\frac 12| H|^2
 + \frac {\alpha'}4 \text{tr}\Big[ | R^+|^2 - |F|^2\Big]&=0 ,\\
  e^{2\phi }d\ast e^{-2\phi} F + A\wedge \ast F- \ast F\wedge A  + \ast H \wedge F&=0,\\
    d\ast e^{-2\phi}H  &=0.
  \end{aligned}
 \end{equation} 
  The full action with fermions is invariant under supersymmetry, acting on the fermions as 
\begin{equation}
 \begin{aligned}{}
      \delta \psi_\mu &=\nabla_\mu^- \varepsilon , \\
    \delta \lambda &= -\frac 12 \gamma\big(d\phi -\frac 1{12}H\big)\varepsilon,\\
    \delta \chi &= -\frac 14 \gamma(F) \varepsilon,
 \end{aligned}
\end{equation}  
 where $\psi$ is the gravitino, $\lambda$ the dilatino, and $\chi$ the gaugino. The quantization map $\gamma$ is explicitly
 \begin{equation}
   \gamma \Big(\frac 1{p!} \omega_{\mu_1\dots \mu_p} dx^{\mu_1}\wedge\dots\wedge dx^{\mu_p}\Big)=  \omega_{\mu_1\dots \mu_p} \gamma^{\mu_1}\dots \gamma^{\mu_p},
 \end{equation} 
 where we use the convention $\{ \gamma^\mu,\gamma^\nu\} = 2g^{\mu\nu}$.  
 The requirement that these variations vanish ensures that a background preserves supersymmetry, and is precisely the set of equations \eqref{IntroSUSYEqtns}. Here the connections $\nabla^\pm$ are related to the Levi-Civita connection of $g$ via
   \begin{equation}
     (\Gamma^-)^a_{bc} = \Gamma^a_{bc} + \textstyle{\frac 12} {H^a}_{bc}, \qquad 
     (\Gamma^+)^a_{bc} = \Gamma^a_{bc} - \textstyle{\frac 12} {H^a}_{bc}.
  \end{equation} 
 In addition to the equations of motion or supersymmetry equations, one has to impose the Bianchi identity
 \begin{equation}
    dH = \frac {\alpha'}4 \text{tr}\Big(R^+ \wedge R^+ - F\wedge F\Big).
  \end{equation}
 It has been proposed to choose the same connection everywhere in the equations, instead of $\nabla^+$ and $\nabla^-$, but here we
stick to the usual convention, which seems to be preferred from a string theoretical point of view  \cite{BeckerSethi09}. We cannot expect the
equations of motion \eqref{EOMs} to be implied by the supersymmetry equations then, as this would require taking into
account all $\alpha'$ corrections. Those equations however which do not involve the gauge field will be satisfied (and the Yang-Mills equation for $F$ as well). They are the $H$-equation $d\ast e^{-2\phi }H=0$ and the
following combination of dilaton equation and trace of the Einstein equation:
   \begin{equation}\label{EinsteinDilatonCombinoalpha}
  \text{Scal} - 8|d\phi|^2 + 6\Delta\phi + \frac 12 |H|^2 =0.
 \end{equation}
 From this we can derive a simple no-go theorem. Suppose the dilaton is constant, then
 \begin{equation}
     \text{Scal} =- \frac 12 |H|^2 ,
 \end{equation} 
 and the scalar curvature must be non-positive. It should be mentioned that there is also a constraint on the cohomology
class of $H$:
 \begin{equation}
  [H] \in H^3\big(M,\, 4\pi^2\alpha' \mathbb Z\big),
 \end{equation} 
 if $dH=0$, which leads to the quantization of the level in WZW models for instance. For $dH\neq 0$ the requirement 
will be that a certain combination of $H$ and the Chern-Simons forms of $R^+$ and $F$ defines an integer cohomology
class, but we will simply ignore this condition in what follows, as most of the spaces we consider have $H^3(M,\mathbb
Z)=0$ anyway. 

\bigskip

 As mentioned in the introduction, in this paper we will solve the BPS equations \eqref{IntroSUSYEqtns} together with the Bianchi identity \eqref{IntroBianchi}, and ignore the equations of motion completely. 

\section{Homogeneous vector bundles}\label{sec_HomVecBdls}
 Let $G$ be a connected compact simple Lie group equipped with the bi-invariant Riemannian metric $g$ (induced
 by minus the Killing form on its Lie algebra $\mathfrak g$),  and $K$ a
naturally reductive subgroup. This means we have an orthogonal splitting of the Lie algebra $\mathfrak g = \mathfrak
k\oplus \mathfrak m$, with ad$(\mathfrak k)\mathfrak m \subset \mathfrak m$, so that $\mathfrak m$ carries a
representation of $\mathfrak k$. Let $(V,\rho)$ be a representation of $K$, and $E= G\times _K V$ the associated
vector bundle over $G/K$, which consists of equivalence classes $[g,v]$ with $g\in G$ and $v\in V$, and
identification $[g,v]=[gk^{-1},\rho(k)v]$ for all $k\in K$.
Its sections are in a 1-1 correspondence with maps $f:G\rightarrow V$ satisfying
 \begin{equation}
   f(gk) = \rho(k)^ {-1} f(g),\qquad \forall k\in K.
 \end{equation} 
 $G$ acts on the space of sections $\Gamma(G/K,E)$ through $(g\cdot f)(h) = f(g^{-1}h)$. The set of $G$-invariant
sections (also called homogeneous sections) is thus given by the constant functions, and therefore in a 1-1
correspondence to the $K$-invariant elements of
$V$:
 \begin{lem}\label{Lem_Hinvariants=GinvSections}
  Let $V$ carry a representation of $K$, then
  $$ \Gamma(G/K, G\times_K V)^G \simeq V^K.$$
 \end{lem}
 We will adopt the following index convention. Basis elements of $\mathfrak k$ will be denoted by $I_k,I_l,\dots$,
those of $\mathfrak m$ by $I_a,I_b,\dots$, and the full set of basis elements of $\mathfrak g$ by $I_\mu,I_\nu,\dots$.
The dual basis of left-invariant 1-forms on $G$ is denoted $e^\mu$, or $e^k$ and $e^a$ for those dual to $I_k$ and
$I_a$. The pull-backs of these forms to $G/K$ will be denoted by the same symbols, and they satisfy the Maurer-Cartan
equations
\begin{equation}
 \begin{aligned}{}
   de^k &= -\frac 12 f^k_{lm}e^l\wedge e^m - \frac 12 f^k_{ab} e^ a\wedge e^b , \\
   de^a &= -\frac 12 f^a_{bc}e^b\wedge e^c - f^a_{bk} e^ b\wedge e^k,
 \end{aligned}
\end{equation}  
 where $f^\lambda_{\mu\nu}$ are the structure constants of $\mathfrak g$, defined by $[I_\mu,I_\nu]=f^\lambda_{\mu\nu}
I_\lambda$. 
Our metric on $\mathfrak g$ will be minus the Killing form
 \begin{equation}
   g(X,Y) =  \text{tr}_\mathfrak g \big(\text{ad}(X)\circ \text{ad}(Y)\big), \qquad X,Y\in\mathfrak g,
 \end{equation}  
 or in coordinates
 \begin{equation}
   g_{ab} = -\big(f^d_{ac}f^c_{bd} +2 f^k_{ac}f^c_{bk}\big) ,\qquad g_{kl}=- \big(f^n_{km}f^m_{ln} +  f^b_{ka}f^a_{lb}\big).
 \end{equation}  
 $g$ is $\mathfrak g$-invariant, thus also $\mathfrak k$-invariant, and gives rise to a homogeneous metric on $G/K$.  

\paragraph{The 3-form.}
 Another important example of a homogeneous section is the following. Define $H \in \Lambda^3 \mathfrak m^*$ through
 \begin{equation}
  H(X,Y,Z ) = -g([X,Y],Z) ,\qquad \forall X,Y,Z\in \mathfrak m,
 \end{equation} 
 or in coordinates
 \begin{equation}
  H=- \frac 16 f_{abc}e^a\wedge e^b\wedge e^c.
 \end{equation} 
 Then $H$ is $K$-invariant, and gives rise to a 3-form on $G/K$. In case $K$ is chosen trivial, $H$ becomes a generator
of $H^3(G,\mathbb Z)=\mathbb Z$ upon proper normalization of the metric $g$. In general $H$ is a natural candidate for
the 3-form of heterotic string theory. For the Bianchi identity we need to know the derivative of $H$, and for its equation of motion $d* H$ (using the notation $e^{abcd}=e^a\wedge e^b\wedge e^c \wedge e^d$):
 \begin{lem}\label{Lem_Hderivatives}
 We have 
 \begin{equation}
  \begin{aligned}{}
   dH= -\frac 14 f_{kab}f^k_{cd} &e^{abcd} = \frac 14 f_{abe}f^e_{cd} e^{abcd}, 
  \end{aligned}
 \end{equation} 
 whereas $d\ast H=0$.
\end{lem}
\begin{proof}
 From the Maurer-Cartan equation we have
 $$ dH = \frac 14 f_{abf}f^f_{cd} e^{abcd} + \frac 14 f_{abc}f^a_{dk} e^{dkbc}.$$
 Consider the last term. It follows from the Jacobi identity that $f_{abc}f^a_{dk}$ splits into a part which is symmetric in $b$ and $d$, and another part symmetric in $c$ and $d$. Therefore this term vanishes. Using again a Jacobi identity, we conclude that
 $$  f_{abe}f^e_{cd} e^{abcd}= - f_{kab}f^k_{cd} e^{abcd}.$$
 $d\ast H$:  We assume the $I_\mu$ to form an othonormal basis, such that $f_{\mu\nu\lambda}$ is totally antisymmetric.
Furthermore we will not keep track of whether an index is up or down, but rather sum over any index appearing more than
once. Then we have
  $$ \ast H= -\frac 1{6(n-3)!} \varepsilon^{a_1\dots a_n} f_{a_1a_2a_3} e^{a_4 \dots a_n},$$
 with derivative
\begin{equation}
 \begin{aligned}\label{Cosets:H_coclosed}
  d\ast H &= \frac {1}{12(n-4)!}  \varepsilon^{a_1\dots a_n} f_{a_1a_2a_3} f^{a_4}_{bc} e^{bc a_5\dots a_n} \\
        & +\frac 1{6(n-4)!} \varepsilon^{a_1\dots a_n} f_{a_1a_2a_3} f^{a_4}_{bk} e^{bk a_5\dots a_n}.
 \end{aligned}
\end{equation}
 The first term is easily seen to vanish: $b$ and $c$ only run over the values of $a_1,a_2$ and $a_3$, giving contributions of the type
 $$ \varepsilon^{a_1\dots a_n} f_{a_1a_2a_3} f_{a_4 a_1a_2}e^{a_1a_2 a_5\dots a_n}, $$ 
 where the two $f$ factors are symmetric in $a_3$ and $a_4$, and thus vanish. Now let us consider the second contribution in \eqref{Cosets:H_coclosed}. We have
\begin{equation*}
 \begin{aligned}{}
   d \ast H &= \frac 1{2(n-4)!} \varepsilon^{a_1\dots a_n}  f_{a_1a_2a_3} f^{a_4}_{a_3k} e^{a_3k a_5\dots a_n}  \\
    &= \frac 1{2(n-3)!} \varepsilon^{a_1\dots a_n}  f_{a_1a_2\mu} f^{a_4}_{\mu k} e^{a_3k a_5\dots a_n}, 
 \end{aligned}
\end{equation*} 
 which vanishes due to the Jacobi identity again.
\end{proof}

\paragraph{Connections.}
 Due to the identification $T^*(G/K)= G\times _K\mathfrak m^*$, a connection on a homogeneous vector bundle $G\times _K
V$ can be considered as a map 
 $$ \nabla : C^\infty(G,V)^K \rightarrow C^\infty(G,V\otimes \mathfrak m^*)^K,$$
 satisfying additional properties. 
The simplest example is the so-called canonical connection $\nabla^-$, acting as
 \begin{equation}
   \nabla_X^- f = X_L (f) \qquad \forall X\in \mathfrak m,\ f\in C^\infty(G,V)^K.
 \end{equation} 
 Here $X_L$ is the left-invariant vector field on $G$ corresponding to $X$. In a trivialization of $T(G/K)$ induced by a
local map $G/K
\rightarrow G$, which allows to pull back the left-invariant 1-forms on $G$ to locally-defined 1-forms on $G/K$, its
connection form is given by
 \begin{equation}
    \Gamma^- =  d\rho_e(I_k)e^k,
 \end{equation} 
 with $d\rho_e$ the differential of $\rho:K\rightarrow $\,Aut($V$) at the identity. As is clear from the definition,
the parallel sections of $\nabla^-$ correspond to constant functions, and thus to $K$-invariant elements of
$V$:
\begin{lem} \label{lem_minusholonomy}
  Let $V$ carry a $K$ representation, then the parallel sections of $G\times_K V$ w.r.t. $\nabla^-$ are in a 1-1
correspondence with $K$-invariant elements of $V$, and by Lemma \ref{Lem_Hinvariants=GinvSections} are precisely the
$G$-invariant sections.
\end{lem}
 As the notation suggests, we will identify $\nabla^-$ with the connection appearing in the gravitino equation 
$\nabla^-\varepsilon =0$, and thereby translate the problem of solving this differential equation into a
representation-theoretical one for the holonomy group $K$. 
 Besides $\nabla^-$, we have some further homogeneous connections on $T(G/K)$, the Levi-Civita connection $\nabla$ of $g$, and the connection $\nabla^+$, given by 
 \begin{equation}
  \begin{aligned}{}
  {\Gamma^-}^c_{ab}  &= \Gamma^c_{ab} + \frac 12 {H^c}_{ab}, \\
  {\Gamma^+}^c_{ab}  &= \Gamma^c_{ab} - \frac 12 {H^c}_{ab}. 
  \end{aligned}
 \end{equation} 
 Explicitly, one finds \cite{LNP10}:
\begin{equation} 
\begin{aligned}
  \label{secCosets_Gamma-}
   \Gamma &= \Big(f^a_{kb} e^k + \textstyle{\frac 12} f^a_{cb}e^c\Big)\big(I_a\otimes e^b\big),\\
    \Gamma^ - &= f^a_{kb} e^k \big(I_a\otimes e^b\big)   , \\
    \Gamma^ + &= \big(f^a_{kb} e^k + f^a_{cb}e^c\big) \big(I_a\otimes e^b\big) =f^a_{\mu b }e^\mu\big(I_a\otimes
e^b\big).
 \end{aligned}
\end{equation}
 The structure group for these connections is generically SO($\mathfrak m)$, but $\nabla^-$ has
structure group $K\subset $ SO($\mathfrak m)$. Their curvatures, as elements of End($\mathfrak m)\otimes
\Lambda^2\mathfrak m^*$ and in coordinates, are
 \begin{equation}\label{Cosets:canCurvature1}
 \begin{aligned}
  R^+ &= -\text{ad}(I_a)\circ\pi_\mathfrak k\circ \text{ad}(I_b) e^a\wedge e^b, \qquad
    {(R^+)^c}_{dab} = 2f^c_{k[a}f^k_{b]d},\\
  R^-&=-\frac 12 f^k_{ab}\, \text{ad}_\mathfrak m (I_k) e^a\wedge e^b,\qquad\qquad \ \ {(R^-)^c}_{dab} =-f^k_{ab}
f^c_{kd},
\end{aligned}\end{equation}
 with $\pi_\mathfrak k:\mathfrak g\rightarrow \mathfrak k$ the orthogonal projection.
For the Bianchi identity we need to know tr$(R^+\wedge R^+)$ and possibly tr$(R^-\wedge R^-)$. These are given by
\begin{equation}  
\begin{aligned}{}
     \text{tr}_\mathfrak m( R^-\wedge R^-) &= \frac 14\langle I_k,I_l\rangle_\mathfrak m f^k_{ab}f^l_{cd}  e^{abcd}, \\
     \text{tr}_\mathfrak m(R^+\wedge R^+ )&= -\frac 14\langle I_k,I_l\rangle_\mathfrak k f^k_{ab}f^l_{cd}e^{abcd},
 \end{aligned}
\end{equation}
 where we introduced the (negative) Killing form $\langle\cdot.\cdot\rangle_\mathfrak k$ of the subalgebra $\mathfrak k$, and similarly
  $$ \langle I_k,I_l\rangle_\mathfrak m = \text{tr}_\mathfrak m\big(\text{ad}(I_k)\circ \text{ad}(I_l)\big).$$
 Using the result of Lemma \ref{Lem_Hderivatives} we conclude that
 \begin{equation}\label{BianchiSolved}
   \text{tr}\big(R^+\wedge R^+ - R^-\wedge R^-\big) = dH,
 \end{equation} 
 which is almost the Bianchi identity. It looks however as if we need to put $\alpha'=4$ to solve the Bianchi identity, 
but this is due to our arbitrary normalization of the metric on $G/K$. Note that the lhs. of \eqref{BianchiSolved} is
completely scale-independent, whereas the rhs. scales with the same factor as the metric. Therefore the Bianchi identity
really fixes the scale in terms of $\alpha'$.

\paragraph{Spinors.}
 The spin bundle on a homogeneous manifold is constructed as follows. Ad-invariance of the Killing form implies that
$\mathfrak k$ acts
orthogonally on $\mathfrak m$, giving rise to an embedding
 \begin{equation}
   \text{ad}_\mathfrak m: \mathfrak k \rightarrow  \mathfrak{so}(\mathfrak m),
 \end{equation} 
 which can be composed with the spin representation $dS: \mathfrak{so}(\mathfrak m)\rightarrow
\mathfrak{spin}(\mathfrak m)$, to give $\widetilde{\text{ad}}:= dS\,\circ \,$ad$_\mathfrak m$. We assume
that this lifts to a representation of $K$, a sufficient condition for this being that $K$ is simply-connected. 
Denoting the spinor
space over $\mathfrak m$ by $S$ (also $S(\mathfrak m)$ occasionally), we get an associated bundle 
 \begin{equation}
   \mathcal S = G\times_K S,
 \end{equation} 
 which is the spinor bundle over $G/K$. The connections we considered before give rise to connections on $\mathcal S$,
and the parallel sections w.r.t. $\nabla^-$ correspond to $K$-invariant elements of $S$. To determine whether there
exist parallel spinors we thus need to know whether the trivial representation of $K$ (or $\mathfrak k$) occurs
in the decomposition of the spinor representation $S$ over $\mathfrak m$ into irreducibles, which is a
purely algebraic task. 

\bigskip

 Suppose then that $\varepsilon$ is parallel w.r.t. $\nabla^-$, so that also $R^-\varepsilon =0$. Then it follows from the symmetry property $R^-_{abcd}=R^-_{cdab}$ that $R^-$ annihilates $\varepsilon$ under the Clifford action as well, $\gamma(R^-)\varepsilon=0$, which makes $R^-$ a candidate for $F$ solving the gaugino equation $\gamma(F)\varepsilon=0$. We have seen before that it is also an excellent candidate to solve the Bianchi identity.  

\bigskip

 The following commutation relation between the quantized 3-form and elements of $\mathfrak k$ acting on spinors over $\mathfrak m$ will be useful:
\begin{lem}\label{Lem_HkComRel}
 For $X\in \mathfrak k$ we have
 \begin{equation}
   [\gamma(H),\widetilde{\text{ad}}(X)]=0.
 \end{equation} 
\end{lem}
\begin{proof}
 A simple calculation in the Clifford algebra shows that
 \begin{equation}
   [\gamma(H),\widetilde{\text{ad}} (X)] = -3{\text{ad}(X)^a}_b f_{acd}\gamma^{bcd},
 \end{equation} 
 but this is proportional to $\gamma(\text{ad}(X)\cdot H)$, where $\cdot $ denotes the action of $\mathfrak{so}(\mathfrak m)$ on $\Lambda^3 \mathfrak m^\ast$, and we know that $H$ is invariant under this action of $\mathfrak k$.
\end{proof}
 This means that $\gamma(H)$ leaves the set of invariant elements in $S(\mathfrak m)$ invariant, so if there is only one invariant spinor, $\gamma(H)$ maps it to a multiple of itself.

\paragraph{The dilaton.}
 In the supergravity equations only the differential $d\phi$ occurs, and if we impose homogeneity again it gives rise
 to a $K$-invariant element of $\mathfrak m^*$. Often there do exist $K$-invariant elements
in $\mathfrak m^*$ if the rank of $K$ is smaller than the rank of $ G$, which correspond to Cartan generators
orthogonal to $\mathfrak k$, but the associated 1-forms on $G/K$ are not exact, and therefore not suitable for our
purpose. We have to conclude that $d\phi=0$
is the only admissible solution for the dilaton.\\
\indent  On the other hand we have seen that a vanishing dilaton is not compatible with positive scalar curvature, which is why we will have to introduce a linear dilaton on an additional $\mathbb R^p$ factor of the total
manifold to obtain a solution of all the supergravity equations. 

\paragraph{Symmetric spaces.}
 Suppose $G/K$ is symmetric, meaning that $[\mathfrak m,\mathfrak m]\subset \mathfrak k$. Then $H=0$, and from the dilatino equation we also have that $d\phi=0$. The equation $\nabla\eta=0$ then tells us that there is a parallel spinor, implying that $M$ is Ricci-flat \cite{Agr06}, which is impossible for symmetric spaces with $G$ semisimple. Thus there are no solutions for symmetric spaces. \\
 \indent On the other hand, it is far from obvious to me why the relation $[\mathfrak m,\mathfrak m]\subset \mathfrak k$ implies that the trivial $\mathfrak k$-representation does not occur in $S(\mathfrak m)$, and a purely Lie algebraic proof would be desirable. 

\section{Spinors on cosets}\label{sec_Spinors}
\paragraph{Representation-theoretic method.}
 Given a coset $G/K$ we need to determine whether the spin representation over $\mathfrak m=\mathfrak g/\mathfrak k$
contains the trivial $\mathfrak k$-representation as an irreducible component.
Suppose for the moment that $\mathfrak k$ is
simple, and denote the set of weights of $\mathfrak m$ by $\Omega(\mathfrak m)$, whereas $\Omega^+(\mathfrak m)$
contains only
the positive ones. Then the weights that appear in $S(\mathfrak m)$ are of the form
 \begin{equation}
  \Omega(S(\mathfrak m)) = \bigg\{ \frac 12\sum_{\alpha\in \Omega^+(\mathfrak m)} \pm \alpha \bigg\},
 \end{equation}
 where all combinations of signs appear (this can be understood from Lemma \ref{Lem_gaugedCartanAlgActionSpinors}
below). Now only the dominant weights can be highest weights of an irreducible representation of $S(\mathfrak m)$,
so it is often enough to determine all the dominant weights in $\Omega(S(\mathfrak m))$. Then a couple of situations can
occur. If the zero weight is not in $\Omega(S(\mathfrak m))$, then the trivial representation is not contained. If it is, one can sometimes conclude by dimensional reasoning that the trivial representation must or must not occur as a component. Sometimes the situation is even simpler:

\bigskip

 Consider the coset SO($n+1)/$SO$(n)=S^{n}$. In this case $\mathfrak m$ is simply the fundamental representation of
$\mathfrak {so}(n)$, and it follows that $S(\mathfrak m)$ is the (Dirac) spinor representation, which is often
reducible, but does not have invariant elements. Thus $S^n$ with its standard round metric does not admit a homogeneous
solution, in accordance with our general result for symmetric spaces.

\bigskip
 
 Despite its elegance we will not employ the representation-theoretic method in the following, but use the more 
down-to-earth approach explained in the next paragraph. The reason for that is that the latter method allows us to
determine how the three-form $H$ acts on invariant spinors, and thereby how to choose the dilaton appropriately to solve
also the dilatino equation $\gamma(d\phi-\frac 1{12} H)\varepsilon =0$. A drawback is that the method is not always
applicable, as explained below.

\paragraph{Direct method for lower-rank subgroups.}
 Recall that the action of $\mathfrak k$ on $\mathfrak m$ defines an embedding $\mathfrak k\subset
\mathfrak{so}(\mathfrak m)$. 
Here we give an explicit construction of the spinor space for the case that $\mathfrak k\subset \mathfrak {su}(\mathfrak
m)$, for a well-chosen complex structure on $\mathfrak m$. We assume
for the time being that rk$(\mathfrak k)<$rk$(\mathfrak g)$, although for certain maximal rank subgroups we will be
able 
to generalize our construction. A particular example where this is possible is the case where $G/K$ is a six-dimensional
nearly K\"ahler manifold. 

\bigskip
 
 A further assumption we want to make is that $\mathfrak k$ and $\mathfrak g$ admit a common root space decomposition, i.e. they are of the form
 \begin{equation}\label{CommonRtSpDec}
  \begin{aligned}{}
       \mathfrak g\otimes \mathbb C &= \mathfrak h \bigoplus_{\alpha\in R^+} \big(\mathfrak g_\alpha \oplus \mathfrak
g_{-\alpha}\big),  \\
       \mathfrak k \otimes \mathbb C &= \text{span}\{H_{p+1},\dots, H_r \} \bigoplus _{\alpha \in S^+} \big( \mathfrak
g_\alpha \oplus \mathfrak g_{-\alpha}\big),
  \end{aligned}
 \end{equation}  
 where the Cartan algebra of $\mathfrak g$ is 
 \begin{equation}
  \mathfrak h = \text{span}\{H_1,\dots, H_r \}
 \end{equation} 
 with $r=$ rk$(\mathfrak g)>$ rk$(\mathfrak k)=r-p$, and the set of positive roots of $\mathfrak g$ is denoted by $R^+$, whereas those of $\mathfrak k$ are contained in $S^+\subset R^+$. The complement $\mathfrak m$ of $\mathfrak k$ in $\mathfrak g$ is then given by
   \begin{equation}
  \mathfrak  m \otimes \mathbb C= \text{span}\{H_{1},\dots,H_p\} \bigoplus_{\alpha\in R^+\setminus S^+}  \big( \mathfrak
g_\alpha \oplus \mathfrak g_{-\alpha}\big),
 \end{equation} 
 which is usually not a Lie algebra. Our final assumption is that the roots in $R^+\setminus S^+$ are higher than  
 those in $S^+$ (which roots are highest, i.e. more positive than others, is
a convention, but it is not always possible to choose the roots in this way; typically if
rk$(\mathfrak k)=$ rk$(\mathfrak g)$ then this ordering is not possible). By restriction, we will consider roots of $\mathfrak g$ as roots of $\mathfrak k$ as well.
 Further we assume a basis of root vectors $E_\alpha$ of the $\mathfrak g_\alpha$ and $E_{-\alpha}$ of the $\mathfrak
g_{-\alpha}$
 to be chosen, obeying the commutation relation
   \begin{equation}\label{rootVectorComRel}
   [E_\alpha,E_\beta] =\begin{cases}
                         2\sum_j\alpha(H_j)H_j & \text{if }\alpha=-\beta \\
             N_{\alpha\beta} E_{\alpha+\beta} & \text{if }\alpha+\beta \text{ is a root}\\
                         0 & \text{otherwise},
                       \end{cases}
 \end{equation}
 where $N_{\alpha\beta}$ are constants, and the normalization can be chosen as \cite{CFT}
 \begin{equation}
   g(H_i,H_j)=\delta_{ij},\qquad g(E_\alpha,E_\beta) = 2\delta_{\alpha,-\beta},\qquad g(H_i,E_\alpha)=0.
 \end{equation}
 Now extend the metric of $\mathfrak m$ to one on $\mathfrak m':=\mathfrak m \oplus \mathbb R^p$, by using the standard metric on $\mathbb R^p$. Furthermore we extend the action of $\mathfrak k$ on $\mathfrak m$ trivially to one on $\mathfrak m\oplus \mathbb R^p$, so that we still have
 \begin{equation}
  \mathfrak k \subset \mathfrak{so}(\mathfrak m'),
 \end{equation}  
 and the spin representation $dS: \mathfrak {so}(\mathfrak m') \rightarrow \mathfrak{spin}(\mathfrak m')$ embeds
 $\mathfrak k$ into the spin algebra. A complex structure $J$ on $\mathfrak m'$ is obtained by defining holomorphic
vectors to be given by the positive root vectors $E_\alpha$, $\alpha \in R^+\setminus S^+$, as well as the following
combinations of Cartan vectors $H_j$ and standard basis vectors $\partial_j$ of $\mathbb R^p$:
 \begin{equation}
   E_j = H_j - i \partial_j,\qquad j=  1,\dots,p, 
 \end{equation} 
 whereas anti-holomorphic ones are given by the negative root vectors $E_{-\alpha}$ and the $ \overline E_j = H_j + i
\partial_j$. The $H_j$ are assumed to be elements of the real Lie algebra $\mathfrak g$, so that roots w.r.t. them are
purely imaginary. From the commutation relation \eqref{rootVectorComRel} and our ordering of the roots
 we deduce the following relations 
 \begin{equation}\begin{aligned}
  \label{sec:Cosets_CplStr}
  [\mathfrak k,\mathfrak m^+]&\subset \mathfrak m^+,\qquad\quad [\mathfrak k,\mathfrak m^-]\subset \mathfrak m^-, \\
  [\mathfrak m^+,\mathfrak m^+] &\subset \mathfrak m^+,\qquad [\mathfrak m^-,\mathfrak m^-]\subset \mathfrak
m^-,
 \end{aligned}   \end{equation}
 where $\mathfrak m^+$ consists of holomorphic vectors in $\mathfrak m'\otimes \mathbb C$, and $\mathfrak m^-$ of 
anti-holomorphic ones. Thus our subalgebra $\mathfrak k$ commutes with the complex structure, and is contained
in the unitary subalgebra of $\mathfrak{so}(\mathfrak m')$:
 \begin{equation}
   \mathfrak k \subset \mathfrak u(\mathfrak m',J).
 \end{equation} 
 It follows from this that the complex structure $J$ extends to an almost-complex structure on the manifold
$\mathbb R^p\times G/K$, which is just a different way of saying that its structure group is contained in $U(\mathfrak m')$. 
Let $e^\alpha$ be the holomorphic 1-forms dual to the $E_\alpha$, and $e^j$ be dual to $E_j$, and denote by $W$ the 
space of all holomorphic 1-forms on $\mathfrak m'$. Then the spinor space is the exterior algebra 
 \begin{equation}
   S=  \Lambda W.
 \end{equation}  
 The Clifford algebra Cl$(\mathfrak {m'}^*)$ is the quotient of the tensor algebra over ${\mathfrak m'}^*$ by the relation 
 \begin{equation}
    v\otimes w - w\otimes v = 2g^{-1}(v,w), \qquad \forall v,w \in {\mathfrak m'}^*,
 \end{equation}   
 where $g^{-1}$ is the metric induced by $g$ on the dual ${\mathfrak m'}^*$. As a vector space the Clifford algebra is isomorphic to the exterior algebra $\Lambda \mathfrak {m'}^*$, but not as an algebra of course. The vector space isomorphism is just the quantization map $ \gamma$. Cl$(\mathfrak {m'}^*)$ acts on $S$ as follows; associate to every
holomorphic 1-form $e^\alpha$ an operator $\zeta^\alpha:=\gamma(e^\alpha)$, an operator $\overline
\zeta^\alpha=\gamma(e^{-\alpha})$ to every anti-holomorphic 1-form $e^{-\alpha}$, as well as $\zeta^j=\gamma(e^j)$ to
$e^j$ and $\overline \zeta^j=\gamma(\overline e^j)$ to $\overline e^j$. They act on $S=\Lambda W$ as creation and
annihilation operators:
\begin{equation}
 \begin{aligned}{}
  \zeta^a \cdot (w^{b_1}\wedge \dots \wedge w^{b_r}) &= w^a \wedge w^{b_1} \wedge \dots \wedge w^{b_r},\\
  \overline \zeta^a \cdot( w^{b_1}\wedge \dots\wedge w^{b_r}) &= \sum_{i=1}^ r (-1)^{i-1} \delta^{a b_i}  w^{b_1} \wedge
\dots\wedge \breve w^{b_i} \wedge \dots \wedge w^{b_r}, 
 \end{aligned}
\end{equation}
 where $\breve w^{b_i}$ means leaving out the element, and the $w^a$ are elements of $W$. This action can be uniquely extended to one of the full Clifford algebra. Inside the Clifford algebra Cl($2n):=$ Cl($\mathbb R^{2n})$ sits the spin algebra $\mathfrak{spin}(2n)$, which is the image of the map
 \begin{equation}
  dS : \mathfrak{so}(2n) \rightarrow \text{Cl}(2n),\ A\mapsto \frac 14 g_{ac}{A^c}_b \gamma^a \gamma^b,
 \end{equation}
 where $\gamma^a = \gamma(dx^a)$.
  \ One can check that this action of the orthogonal Lie algebra restricts to the usual action of $\mathfrak {su}(n)$ on $\Lambda \mathbb C^{n*}$, in particular the decomposition
 $$ S =\oplus_l S^l =\mathbb C \oplus \mathbb C^{n*} \oplus \Lambda^2 \mathbb C^{n*} \oplus \dots \oplus  \Lambda^n \mathbb C^{n*}$$
 is preserved by this subalgebra, and then even by $\mathfrak u(n)$, although the additional $\mathfrak u(1)$-component acts in a non-standard way. To show that $\mathfrak k$ leaves at least one spinor invariant, it is therefore enough to show that it is contained in $\mathfrak{su}(\mathfrak m')$, as this acts trivially on $S^0\oplus S^n$. We can easily determine the action of the Cartan generators of $\mathfrak k$ on spinors (recall $\widetilde{\text{ad}}(X)=$ dS $\circ $ ad$(X)$ for $X\in \mathfrak k$).
 \begin{lem}\label{Lem_gaugedCartanAlgActionSpinors}
  We have for the $H_j\in \mathfrak k$, $j={p+1},\dots,r$:
    \begin{equation}
  \widetilde{\text{ad}}(H_j) = \frac 12 \sum_{\alpha\in R^+\setminus S^+} \alpha(H_j)\big( \overline \zeta^\alpha \zeta^\alpha - \zeta^\alpha \overline \zeta^\alpha\big).
 \end{equation} 
  \end{lem}
 \begin{proof}
 According to the definition of $dS$ we get
  \begin{align*}
   \widetilde{\text{ad}}(H_j) & =\frac 14 g_{ac}  {(H_j )^c}_b\gamma^a\gamma^b  \nonumber\\
  &= \frac 14 g_{\alpha\overline \gamma} {(H^j)^{\overline\gamma}}_{\overline \beta} \zeta^\alpha \overline \zeta^\beta + \frac 14 g_{\overline \beta \gamma} {(H^j)^{\gamma}}_{\alpha}  \overline \zeta^\beta\zeta^\alpha,
 \end{align*}
 where we introduced $[H_j,E_\alpha]=:{(H_j)^\beta}_\alpha E_\beta$. But this just means that ${(H_j)^\beta}_\alpha= {\delta^\beta}_\alpha \alpha(H_j)$, and ${(H_j)^{\overline \beta} }_{\overline \alpha} = -{\delta^\beta}_\alpha \alpha(H_j)$. Furthermore $g_{\alpha,\overline \beta}=g_{\alpha,-\beta}=2\delta_{\alpha\beta}$.
 \end{proof}

\begin{prop}
 Under the assumptions of this section suppose that 
 \begin{equation}\label{RootSumZeroCond}
  \sum_{\alpha \in R^+\setminus S^+} \alpha =0
 \end{equation} 
 as roots of $\mathfrak k$. Then $\mathfrak k$ acts trivially on the completely empty and completely filled states 
 $$ S^0 \subset S,\qquad \text{and} \qquad S^{m} \subset S,$$
 where $m=\dim(\mathfrak m')$.
\end{prop}
\begin{proof}
 We know already that $\mathfrak k\subset \mathfrak u(\mathfrak m')$ leaves invariant the decomposition $S=\oplus_l
S^l$. 
The subrepresentations on $S^0$ and $S^m$ are 1-dimensional, and according to Lemma
\ref{Lem_gaugedCartanAlgActionSpinors} the Cartan elements of $\mathfrak k$ act by multiplication with
  $$ \widetilde{\text{ad}}(H_j)\Big|_{S^0} = \frac 12 \sum_{\alpha \in R^+\setminus S^+} \alpha(H_j), \qquad \widetilde{\text{ad}}(H_j)\Big|_{S^{m}} 
= -\frac 12 \sum_{\alpha \in R^+\setminus S^+} \alpha(H_j).$$
 If these vanish, the representations are trivial.
\end{proof}
 A closer inspection shows that the trivial representation occurs with multiplicity at least $2^{p+2}$, as $\mathfrak k$
 acts trivially on all of the holomorphic vectors $E_j$. The requirement that
the supersymmetry generator in heterotic supergravity be Majorana-Weyl will reduce the amount of supersymmetry
preserved in the space orthogonal to $\mathbb R^p\times G/K$ to $\mathcal N=2^{p}$, whereas the dilatino equation
reduces
it even further.

\bigskip

A different way to write the condition \eqref{RootSumZeroCond} is in terms of the Weyl vector $\rho=\frac 12
\sum_{\alpha\in R^+}\alpha$. Then it reads
 \begin{equation}
  \rho_{\mathfrak g} \big|_{\mathfrak k}=\rho_\mathfrak k.
 \end{equation} 
 As we mentioned already, it is equivalent to $\mathfrak k\subset \mathfrak{su}(\mathfrak m)$.
 The condition cannot be satisfied by a
 subalgebra of maximal rank, as this would imply $\sum_{\alpha \in R^+\setminus S^+} \alpha =0$ as roots of
$\mathfrak g$, contradicting the positivity of all the roots in $R^+\setminus S^+$. In this case one can sometimes
choose the complex structure in a different way, i.e. such that holomorphic vectors do not coincide with positive root
vectors, and apply the procedure presented here analogously. This will be demonstrated for $G_2/$SU(3) below.

\bigskip

 Summarizing the results of this section, we have proven that given a naturally reductive subalgebra 
$\mathfrak k\subset \mathfrak g$ of lower rank, admitting a common root space decomposition with $\mathfrak g$ of the
form \eqref{CommonRtSpDec} such that the roots not belonging to $\mathfrak k$ are all higher than those of $\mathfrak
k$, and such that their sum acts trivially on the Cartan generators of $\mathfrak k$, it admits invariant spinors over
$\mathfrak m\oplus \mathbb R^p$, where $p= $ rk($\mathfrak g)-$ rk($\mathfrak k$). \\
 \indent With Lemma \ref{lem_minusholonomy} we conclude that there are invariant, and thereby $\nabla^-$-parallel spinors over $\mathbb R^p \times G/K$. Furthermore we have a very explicit construction of these spinors, which will allow us to solve also the dilatino
equation in this case.
 
\paragraph{The quantized three form.}
 For the dilatino equation we also need to know how the 3-form $H$ acts on the invariant spinors, i.e. we have to 
determine $\gamma(H)\varepsilon$ for $\varepsilon \in S^0 \oplus S^ m$. Using the commutation relations
\eqref{rootVectorComRel} we can determine $H$:
 \begin{equation}\label{3formexplct}
  \begin{aligned}{}
        H&= -\frac 16 f_{abc} e^{abc} \\
  &= 2\sum_{\underset{j=1,\dots,p}{\alpha \in R^+\setminus S^+}} \alpha(H_j) e^{-\alpha}\wedge e^{\alpha} \wedge h^j  \\
  &- 2\sum_{\underset{\scriptstyle{\alpha<\beta}}{\alpha,\beta\in R^+\setminus S^+}} \Big[N_{\alpha\beta} e^{\alpha}\wedge e^\beta\wedge e^{-(\alpha+\beta)} +N_{-\alpha,-\beta} e^{-\alpha}\wedge e^{-\beta}\wedge e^{\alpha+\beta}\Big], 
  \end{aligned}
 \end{equation}
  with $h^j$ dual to $H_j$. Under the quantization map $\gamma $ the last line eliminates $S^0$ and $S^m$, and therefore we only need to consider 
 \begin{equation}
   H'  =  2\sum_{\underset{j=1,\dots,p}{\alpha \in R^+\setminus S^+}} \alpha(H_j) e^{-\alpha}\wedge e^{\alpha} \wedge h^j .
 \end{equation}
 By construction the 1-form $h^j$ is quantized to 
 \begin{equation}
  \gamma(h^j) =  \zeta^j + \overline \zeta^j,
 \end{equation} 
 and we conclude that 
\begin{equation}
   \gamma(H') =6 \sum_{\underset{j=1,\dots,p}{\alpha\in R^+\setminus S^+}} \alpha(H_j)(\overline\zeta^\alpha\zeta^\alpha
- \zeta^\alpha\overline \zeta^\alpha )  (\zeta^j+\overline \zeta^j),
 \end{equation}
 so that
 \begin{align}
   \gamma(H) \Big|_{\mathcal S^0} \ :\  &\ \mathcal S^0 \rightarrow \mathcal S^1,\ 1\mapsto
6\sum_{\underset{j=1,\dots,p}
{\alpha\in R^+\setminus S^+}} \alpha(H_j) w^j ,\nonumber \\
 \gamma(H)\Big|_{\mathcal S^m} \ :\ &\ \mathcal S^m\rightarrow \mathcal S^{m-1}, \\
      w^1&\wedge\dots\wedge w^m\mapsto  6\sum_{\underset{j=1,\dots,p}{\alpha\in R^+\setminus S^+}} (-1)^j\alpha(H_j) 
w^1\wedge \dots\wedge \breve w^j\wedge\dots\wedge w^m.\nonumber
 \end{align}
 Therefore the dilatino equation 
 $ \gamma(d\phi - \sfrac 1{12}H) \varepsilon =0$
 is not satisfied if we choose $\phi=0$ and $\varepsilon \in S^0 \oplus S^ m$. To overcome this problem we can introduce a linear dilaton on the additional $\mathbb R^p$-factor. Take the real valued
 \begin{equation}
   \phi(x) = \frac i{2}\sum_{\underset{j=1,\dots,p}{\alpha\in R^+\setminus S^+}} \alpha(H_j) x^j.
 \end{equation}
 Its differential acts on $\mathcal S$ as
 \begin{equation}
   \gamma(d\phi) = {\frac 12}\sum_{\underset{j=1,\dots,p}{\alpha\in R^+\setminus S^+}} \alpha(H_j) ( \zeta^j-\overline\zeta^j),
 \end{equation}
 which implies
 \begin{equation}
   \gamma(d\phi)\Big|_{\mathcal S^0\oplus \mathcal S^m}  = \frac 1{12} \gamma(H)\Big|_{\mathcal S^0\oplus \mathcal S^m}.
 \end{equation}
 How do we then construct the supersymmetry generator $\varepsilon$? For the whole construction to be useful in 
heterotic string theory, the dimension of $\mathbb R^p \times G/K$ must be less than 10, i.e. we need
 \begin{equation}
  l:= \dim G-\dim K + \text{rk}\,G - \text{rk}\,K <10,
 \end{equation}  
 then we complete this to a 10-dimensional Lorentzian manifold $M= \mathbb R^{9-l+p ,1} \times G/K$. The spinor 
bundle over $M$ is the tensor product of the one over $\mathbb R^{l-1,1}$ and the one over $\mathbb R^p \times G/K$, and
there is a (anti-linear) charge conjugation $C:S\rightarrow S$ mapping $S^0$ to $S^m$ times a complement which does not
affect the
action of forms over $\mathbb R^p \times G/K$. Let $\eta$ be any constant, positive chirality spinor on $\mathbb
R^{l-1,1}$. Then $\varepsilon$ is chosen proportional to $1\otimes \eta + C(1\otimes \eta)$, with $1
\in S^0\simeq \mathbb C$, and this is parallel w.r.t. $\nabla^-$, annihilated by $\gamma(d\phi -\frac 1{12} H)$,
and Majorana-Weyl, as required by heterotic string theory. 

\bigskip

 There is a little sublety related to the amount of supersymmetry preserved. Suppose that the
dimension of $\mathbb R^p\times G/K$ is such that it admits Majorana spinors, and that $\mathcal S^0$ is invariant.
Then there are at least two invariant Majorana spinors, given by 
 \begin{equation}
   1 + C\cdot 1 ,\qquad \text{and} \qquad i(1-C\cdot 1).
 \end{equation} 
 These generate just one supersymmetry however, as two Majorana spinors $\eta + C\eta $ and $i(\eta - C\eta)$ over
$\mathbb R^{l-1,1}$ tensored with our invariant spinors lead to the same set of spinors over the total 10-dimensional
space. 

\bigskip

 What is then the amount of supersymmetry preserved in $\mathbb R^{9-l,1}$? Consider all those
spinors that are
obtained from the completely empty one by acting with an arbitrary number of creation operators $\zeta^j$, $j=1,\dots,
p$. These generate a complex $2^p$ dimensional space $\tilde S=\tilde S^0\oplus \dots \oplus \tilde S^p$, which is
invariant under $\gamma(H)$ and $\gamma(d\phi)$, and we chose
$\phi$ such that $\tilde S^0=S^0$ is annihilated by $\gamma(d\phi - \frac 1{12} H)$. On the higher subspaces $\tilde
S^k$ the condition $\gamma(d\phi - \frac 1{12} H)\psi=0$ becomes one linear equation, so that generically there should
be
$2^p-p$ invariant elements. Indeed one can check that $\tilde S^p$ is not invariant, whereas $\tilde S^2$ has one
invariant spinor, and so on. We conclude that our solutions have 
 \begin{equation}
   \mathcal N\geq  2^p-p
 \end{equation} 
 supersymmetry, with equality if there are no additional invariant spinors over $\mathbb R^p\times G/K$ to those we
constructed explicitly. In
particular for $p=0$ and $p=1$ this leads to $\mathcal N\geq 1$, and these are the cases we shall
consider in the examples below. 
For equal rank subgroups we will have to add a factor $\mathbb R$ by hand on which the linear
dilaton lives, and the orthogonal space will be $\mathbb R^{8-\dim G+\dim K,1}$.

\section{Homogeneous Sasaki-Einstein manifolds}\label{sec_Sasaki}
 This is a particular class of examples where the assumptions on the root space decomposition of $\mathfrak k$ and $\mathfrak g$ are satisfied. A good general reference on Sasaki spaces is the book by Boyer and Galicki \cite{Boyer08}, and a simple definition is to call a Riemannian manifold Sasaki-Einstein, if its metric cone is a Calabi-Yau. Similarly, 3-Sasakian manifolds by definition have hyperk\"ahler cones. According to Theorem 11.1.13 of \cite{Boyer08}, every homogeneous Sasaki-Einstein manifold is the total space of a principal U(1)-bundle over a so-called generalized flag manifold. These latter spaces are by definition cosets of a Lie group $G$ by the centralizer of a torus, and they carry a K\"ahler-Einstein structure \cite{Arvan03}. The classification and an exhaustive list of examples for non-exceptional $G$ can also be found in Arvanitoge\'orgos' book \cite{Arvan03}.
 We now give an explicit method to obtain the U(1)-bundles, which leads to many of the examples considered below. 

\bigskip
 
 Let $ G$ be a again a connected compact simple Lie group with Lie algebra $\mathfrak g$, and Cartan subalgebra $\mathfrak h$. Choose the positive roots, and let $\beta$ be the highest simple root. Denote again by $R^+$ the set of positive roots, and by $S^+\subset R^+$ the set of roots which are linear combinations of all the simple roots except $\beta$. Then $\mathfrak g$ has a root space decomposition of the form
 \begin{equation}
  \mathfrak g \otimes \mathbb C= \mathfrak h \bigoplus_{\alpha \in R^+} \big(\mathfrak g_\alpha \oplus \mathfrak g_{-\alpha}\big), 
 \end{equation}    
 and we define the subalgebras $\mathfrak t\subset\mathfrak c\subset\mathfrak g$:
 \begin{equation}
  \begin{aligned}{}
    \mathfrak c &= \bigg[\mathfrak h \bigoplus_{\alpha \in S^+} \big(\mathfrak g_\alpha \oplus \mathfrak g_{-\alpha}\big)\bigg] \cap \mathfrak g , \\
   \mathfrak t &= \big\{H\in \mathfrak h\ |\ \alpha(H)=0\ \forall \alpha \in S^+\big\} \cap \mathfrak g.
  \end{aligned}
 \end{equation} 
 Furthermore we need the space
 \begin{equation}
   \mathfrak m = \bigg[\mathfrak t \bigoplus_{\alpha \in R^+\setminus S^+}   \big(\mathfrak g_\alpha \oplus \mathfrak g_{-\alpha}\big)\bigg] \cap \mathfrak g .
 \end{equation} 
 These definitions imply the following commutation relations:
 \begin{equation}
  \begin{aligned}{}
    [\mathfrak c,\mathfrak c] &\subset \mathfrak c,\qquad [\mathfrak c,\mathfrak t] =0 , 
   \qquad [\mathfrak c,\mathfrak m] \subset \mathfrak m. 
  \end{aligned}
 \end{equation} 
 Denote by $C,T\subset G$ the corresponding Lie groups (with $T\simeq $ U(1)), then $C$ is the centralizer of $T$ in $G$, and $T\subset C$ is a normal subgroup. Thus $K:=C/T$ is a group again, and its Lie algebra is the orthogonal complement of $\mathfrak t$ in $\mathfrak c$ w.r.t. the Killing form. Due to the relation $(G/K)/T \simeq G/C$ we get a $U(1)$-fibration
 \begin{equation}
   \pi: G/K \rightarrow G/C,
 \end{equation}  
 where the base space carries a homogeneous K\"ahler-structure, with complex structure induced by the choice of positive roots. Let 
 \begin{equation}
   \delta = \frac 12 \sum_{\alpha \in R^+\setminus S^+}\alpha.
 \end{equation} 
 It vanishes on the Cartan generators of $\mathfrak k$. Defining the 1-form $A \in \Omega^1(G/K; \, \mathfrak t$) through
 \begin{equation}
   A = h\otimes H_\mathfrak t,
 \end{equation}  
 with $H_\mathfrak t\in \mathfrak t$ and $h\in \Omega^1(G/K)$ the left-invariant 1-form corresponding to the dual of $H_\mathfrak t$, we get a connection on the bundle $G/K\rightarrow G/C$ \cite{KoNoI}, as $A$ satisfies
 \begin{equation}
   A(H_L) = H_\mathfrak t,\qquad R_t^\ast A = A \ \big(= \text{Ad}(t^{-1})A\big)\qquad \forall t\in T.
 \end{equation} 
 Here $H_L$ is the left-invariant vector field corresponding to $H_\mathfrak t$, and $R_t:G/K\rightarrow G/K,\ [g]\mapsto [gt]$. The curvature $\Omega\in \Omega^2(G/K;\mathfrak t)$
 of this connection is calculated using the Maurer-Cartan equation, and given by:
 \begin{equation}\begin{aligned}{}
   \Omega =dA &=  -2\sum_{\alpha \in R^+\setminus S^+} \alpha(H_\mathfrak t) e^{\alpha}\wedge e^{-\alpha}\, \otimes H_\mathfrak t \\
   & = \frac 2{\delta(H_\mathfrak t)} \sum_{\alpha \in R^+\setminus S^+} (\alpha,\delta) e^\alpha \wedge e^{-\alpha}\,\otimes H_\mathfrak t, 
 \end{aligned} 
\end{equation}  
 with $(\delta,\alpha)$ the inner product on weights induced by minus the Killing form.
 Let $\rho$ be the minimal non-trivial representation of $\mathfrak t$ on $\mathbb C$ which exponentiates to an action
of $T$, it gives rise to an associated line bundle $L\rightarrow G/C$, whose curvature $\Omega^L\in \Omega(G/C, i\mathbb
R)$ is the pull-back of $\rho(\Omega)$ under any set of local sections. We obtain a K\"ahler structure on $G/C$ by the
choice $\omega = \frac i  {2\pi }\,{\Omega^L}$ for the K\"ahler form $\omega$, and according to \cite{Arvan03} the
K\"ahler metric is Einstein. It follows that $L$ is ample, and
 the classification result Theorem 11.1.13 in \cite{Boyer08} (and its proof) imply that $G/K$ carries a Sasaki-Einstein
structure.

\bigskip

 What is more important for us is that the result $\delta |_{\mathfrak k\cap \mathfrak h}=0$ implies that there are
parallel  spinors on $G/K \times \mathbb R$, and all the conditions from our previous section are satisfied. Thus we
find solutions to the heterotic supergravity equations on all of these spaces. It should be mentioned that the
Sasaki-Einstein metric is never the one induced by the Killing form, which we are using. 

\bigskip

 Not all of the examples given below fall into this class however. Besides the Einstein-Sasaki manifolds, we will also 
find solutions on nearly K\"ahler, nearly parallel $G_2$, and 3-Sasakian manifolds. These latter do carry an
Einstein-Sasaki structure as well, but they allow for more parallel spinors, and besides the U(1)-fibration over a
K\"ahler manifold, they also possess an Sp(1)-fibration over a quaternionic K\"ahler manifold. The only example of this
kind we include is Sp$(n+1)/$Sp($n)=S^{4n+3}$, but probably all of them could be used. The
3-Sasakian metric is related to the Killing form in \cite{Bielawski96}, and there is a complete list of the homogeneous
examples in any of \cite{Agr06, Bielawski96, Boyer98, Boyer07, Boyer08}.

\section{Examples}\label{sec_bsp}
 We will treat the cases listed in the introduction, and additionally SU(2) as an example of a group manifold. For the
spaces SU($n+1)/$SU($n$), Sp($n+1)$/Sp$(n)$, Sp$(n)$/SU($n$), SO($2n)$/SU($n)$ and SO$(n+1)$/SO($n-1$) it follows
already from the general discussion about Sasaki-Einstein manifolds above that solutions exist. Nevertheless we give
some details below, which allow us to say more about the amount of supersymmetry preserved and determine it in some
cases, and furthermore to calculate the three-form explicitly. We will not take care to normalize the generators as in
the general discussion above, but otherwise use the same conventions. In particular for a root $\alpha$ we denote by
$e^\alpha$ the 1-form dual to the root vector $E_\alpha$, whereas the dual of a Cartan generator $H_j$ is denoted by
$h^j$. As there will always be just one Cartan generator in $\mathfrak m$, we drop the index and denote its dual simply
by $h$. The generic form of the three-form on the Sasaki spaces is then 
 \begin{equation}
   H = -2\sum_{\alpha\in R^+\setminus S^+} e^\alpha\wedge e^{-\alpha} \wedge h,
 \end{equation} 
 unless relations of the type $\alpha+\beta=\gamma$ exist between three roots in $R^+\setminus S^+$, in which case we
get further contributions from the last line in \eqref{3formexplct}.

\paragraph{SU(2) \& other group manifolds.}
 The simplest examples are the ones with $K$ the trivial group. In this case existence of invariant spinors is trivial 
and need not be checked. The manifold on which our fields live is $\mathbb R^p \times G$, with $p=$ rk$(G$), and the
dilaton is a linear function on $\mathbb R^p$. What is further special about this case is that $H^3(G,\mathbb Z) =
\mathbb Z$, so in order to satisfy the condition $[H]\in H^3(M,4\pi^2 \alpha'\mathbb Z)$ we have to adjust the
scale of $G$, whereas the Bianchi identity is trivially satisfied, with $dH=0$ and $R^+=R^-=0$. 
 For cosets $G/K$ on the other hand the third cohomology may well be trivial, and it is the Bianchi identity
which fixes the scale. 

\bigskip

 The models we obtain here are the low-energy description of Wess-Zumino-Witten models, and well-known.
 SU(2) is the simplest of the group manifolds, with rank 1. Thus we have to consider $\mathbb R\times $SU(2). Choose a
basis $I_1,I_2,I_3$ of the Lie algebra $\mathfrak{su}(2)$ satisfying $[I_i,I_j] = {\epsilon_{ij}}^k I_k$. Then there is
a corresponding basis of left-invariant vector fields on SU(2), which we denote by $I_{1,2,3}$ again, and a dual basis
of left-invariant 1-forms $e^1,e^2,e^3$. 
 The metric on $\mathbb R\times $SU(2) is given by
 \begin{equation}
   g = dx\otimes dx + \sum_j e^j \otimes e^j,
 \end{equation} 
  the $H$-field becomes proportional to the volume form of SU(2):
 \begin{equation}
   H = -e^1\wedge e^2 \wedge e^3,
 \end{equation} 
 and the dilaton is
 \begin{equation}
   \phi(x) = -\frac 12 x.
 \end{equation}  
 Contrary to the Levi-Civita connection, both connections $\nabla^+$ and $\nabla^-$ are flat, i.e. $R^+=R^-=0$, and 
the Bianchi identity $dH=0$ holds. So far all of this remains true for
higher-dimensional groups as well. In this particular case one can check that besides the supersymmetry equations also
the heterotic equations of motion are satisfied (and the higher order corrections in $\alpha'$ vanish). This is not true for simple groups in general, but should not be expected either for
groups of dimension $\geq 8$. The amount of supersymmetry preserved in $\mathbb R^{5,1}$ is $\mathcal N=1$.

\bigskip

 The only condition we did not check so far was the quantization condition on $H$. In order to calculate the cohomology class of $H$, we pull it back to $S^3$ along the map 
 \begin{equation}
  \iota:S^3\rightarrow SU(2),\ (\eta,\xi_1,\xi_2)\mapsto
 \left(
 \begin{array}{cc}
          e^{i\xi_1}\sin(\eta) & e^{i\xi_2}\cos (\eta) \\
         -e^{-i\xi_2}\cos(\eta) & e^{-i\xi_1}\sin(\eta)
\end{array}\right)
 \end{equation}
 \big(where $\eta\in [0,\frac \pi 2),\ \xi_1,\xi_2\in [0,2\pi)$\big) to calculate its integral:
 $$   \int_{SU(2)} H =\int_{S^3}\iota^\ast H=-16 \pi^2.$$
 Thus the cohomology class of $H$ is
  \begin{equation}\label{SU(2)-H-Cohom}
   [H] = -16\pi^2\in H^3\big(\text{SU(2)};\mathbb R\big)\simeq \mathbb R.
  \end{equation}
 To satisfy the integrality constraint we have to rescale the metric as $g'=\mu g$, with
 \begin{equation}
  \mu = \frac {\alpha' n}4,\qquad n\in \mathbb N,
 \end{equation} 
 which automatically rescales $H'=\mu H$ (as $H$ is defined in terms of the metric), and leaves invariant the 
supersymmetry equations. There is thus no continuous volume modulus in the game, as is well-known for WZW models. 
 For higher-dimensional simple groups the analysis is very similar.

\paragraph{SU($\mathbf{n+1)}$/SU($\mathbf{n )=S^{2n+1}}$.}
 Here the difference of the ranks is always 1, so that we get a model based on $\mathbb R\times S^{2n+1}$, where the 
metric on the sphere is different from the round one \cite{Besse87}. A Cartan basis of $\mathfrak{su}(n+1)$ is
given by
\begin{equation}
\lambda_1=\left(\begin{smallmatrix}
 	     i &  && &&  &&\\
             &  -i &&  && &&\\
             &   && 0&& &&\\ 
             & && &&  \ddots &&     \\
             & && && && 0  
   \end{smallmatrix}\right), \quad
\lambda_2=\left(\begin{smallmatrix}
 	     0 && & && && && \\
	     && i &  && &&  &&\\
             && &  -i &&  && &&\\
             && &   && 0&& &&\\ 
             && & && &&  \ddots &&     \\
             && & && && && 0  
   \end{smallmatrix}\right), \dots \ 
 \lambda_n=\left(\begin{smallmatrix}
 	     0 &&  && &&  &\\
             &&  \ddots &&  && &\\
             &&   && 0&& &\\ 
             && && &&  i &     \\
             && && && & -i
   \end{smallmatrix}\right)
\end{equation} 
 with positive root vectors  $E_{ij}$ for $j>i$ and $(E_{ij})_{ab}=\delta_{ai}\delta_{bj}$. The corresponding negative
root vectors are $-E_{ji}$. We have 
 \begin{equation}
  [\lambda_a,E_{ij}] = i[E_{a,a}-E_{a+1,a+1},E_{ij}] = i\big(\delta_{ai}-\delta_{a+1,i}-\delta_{a,j} +
\delta_{a+1,j}\big)E_{ij},
 \end{equation} 
 so that the positive roots with respect to $\lambda_1,\dots,\lambda_n$ are 
 \begin{equation}
   \alpha_{ab} = i(0,\dots,0,\underset{(a-1)}{-1},\underset{(a)}{+1},0,\dots,0,\underset{(b-1)}{+1},\underset
{(b)}{-1},0,\dots,0)_\lambda,
 \end{equation} 
 with $1\leq a<b \leq n +1$ (entries beyond the $n$ components of the root are dropped). The Cartan generators
$\lambda_a$ are not orthogonal though, a better basis is
given by
 \begin{equation}
  H_l = \frac 1{\sqrt{l(l+1)}} \sum_{a=1}^l a\lambda_a,
 \end{equation} 
 or, more explicitly:
  \begin{equation}
  \begin{aligned}\label{su(n+1)Cartanbasis}
   H_1 &= \frac 1{\sqrt 2}\left(\begin {array}{cccc}
 	     i &  & &  \\
             &  -i &   &\\
             &   & & \\ 
             & & &        
         \end {array} \right),\qquad
    H_2 = \frac 1{\sqrt 6}\left(\begin {array}{ccccc}
 	     i & & & & \\
              & i & & &\\
             &   & -2i & & \\
             & &   & &\\ 
             & & &    &    
         \end {array} \right),\ \dots ,\\
     H_n &= \frac 1{\sqrt{n(n+1)}}\left(\begin {array}{cccc}
 	     i &  & & \\
             &  \ddots  & &\\ 
             & &  i & \\ 
             & & & -ni        
         \end {array} \right). 
 \end{aligned}\end{equation}
 The roots w.r.t. these generators are 
 \begin{equation}
\begin{aligned}{}
 \label{su(n+1)roots}
  \alpha_{ab}=i\Bigg(0,\dots,0,\underset{(a-1)} {-\sqrt{\frac{a-1}a}},&\frac 1{\sqrt{a(a+1)}},\frac
1{\sqrt{(a+1)(a+2)}},\dots  \\
  &\dots , \underset{(b-2)}{\frac 1{\sqrt{(b-2)(b-1)}}},\underset{(b-1)}{\sqrt{\frac b{b-1}}},0,\dots,0\Bigg)_H.
\end{aligned}
 \end{equation} 
  The upper left corner $\mathfrak{su}(n)$ subalgebra of $\mathfrak{su}(n+1)$ has 
Cartan generators $H_1,\dots,H_{n-1}$, whereas $H_n$ generates the orthogonal torus. The roots belonging to
$\mathfrak m$ are 
 \begin{equation}
  R^+\setminus S^+ = \{ \alpha_{a,n+1}\ |\ a=1,\dots,n\}.
 \end{equation} 
 Their action on the Cartan generators of $\mathfrak{su}(n)$ is as follows. To every $H_a$ with $1\leq a<n$ we have 
one root $\alpha_{a+1,n+1}$ and $a$ roots $\alpha_{j,n+1}$ with $1\leq j\leq a$ such that
 \begin{equation}
   \alpha_{a+1,n+1}(H_a)=-i\sqrt{\frac a{a+1}},\qquad \alpha_{j,n+1}(H_a) = \frac i{\sqrt{a(a+1)}}.
 \end{equation} 
 Therefore the relation
 \begin{equation}
   \sum_{\alpha \in R^+\setminus S^+} \alpha\ \Big|_{\mathfrak {su}(n)} =0,
 \end{equation} 
 holds, and the conditions of section \ref{sec_Spinors} are satisfied. Thus we can introduce a linear dilaton on
$\mathbb
R$ and find an 
invariant spinor (in both $S^0$ and $S^{n+1}$) satisfying all of the supersymmetry conditions, as well as the Bianchi
identity. The amount of supersymmetry preserved is $\mathcal N=1$, which follows from the fact that
$\mathfrak{su}(n)\subset \mathfrak{spin}(2n)$ has exactly 2 invariant Majorana-spinors (if they exist), and
$\mathfrak{so}(\mathfrak m)= \mathfrak{so}(2n+1)$, whith $\mathfrak{su}(n)$ embedded in the upper left
$\mathfrak{so}(2n)$ in the standard way. As we have no relations between three roots of $\mathfrak m$, the three form is simply
 \begin{equation}
   H  \sim\ \sum_{a=1}^n  e^{\alpha_{a,n+1}} \wedge e^{-\alpha_{a,n+1}}\wedge h, 
 \end{equation} 
 with conventions explained at the beginning of the section.

\paragraph{Sp($\mathbf{n+1}$)/Sp$(\mathbf {n)=S^{4n+3}}$.}
 For the spheres $S^{4n+3}$ we have this alternative representation as cosets, again with rank-difference 1. Here
$\mathfrak{sp}(n)$ denotes the compact real form of the Lie algebra $C_n$, not the symplectic Lie algebra which is a
non-compact real form. 
 The definition is
 \begin{equation}
   \mathfrak{sp}(n)=\{ X \in \mathbb H^{n\times n}\ |\ X^\dagger +X=0\},
 \end{equation} 
 where $\mathbb H$ are the quaternions, and $\dagger$ denotes the (quaternion) conjugate transpose. The quaternion
 conjugate is explicitly given by
 $$ (a+ib+jc+kd)^\ast = a-ib-jc-kd. $$
 A matrix satisfying $X^\dagger +X=0$ has only imaginary entries on the diagonal and the off-diagonal entries are 
completely determined by the upper-triangular part, so that the real dimension of $\mathfrak{sp}(n)$ is $3n +
2n(n-1)=2n^2+n$. A set of Cartan generators is given by the diagonal matrices $H_a$, $a=1,\dots,n$ with entry $i$ at the
$(a,a)$-th position, and zeros everywhere else. In particular the rank is $n$. The root space decomposition takes the
following form. Besides the Cartan generators we have the basis elements
\begin{equation}
 \begin{aligned}{}
    Q_a &=\text{diag}(0,\dots,0,j,0,\dots, 0),\\
  P_a &= \text{diag}(0,\dots,0,k,0,\dots,0), 
 \end{aligned}
\end{equation} 
 for $a=1,\dots,n$, and
 \begin{equation}\begin{aligned}
  E_{rs} &= \left(\begin{array}{cccc}
       0 & & & \\
        & \ddots &1 & \\
        & -1  & \ddots & \\
       & &  & 0 
   \end{array}\right)
 ,\qquad 
  I_{rs} = \left(\begin{array}{cccc}
       0 & & & \\
        & \ddots &i & \\
        & i  & \ddots & \\
       & &  & 0 
   \end{array}\right) , \\
  J_{rs} &= \left(\begin{array}{cccc}
       0 & & & \\
        & \ddots &j & \\
        & j  & \ddots & \\
       & &  & 0 
   \end{array}\right)
 ,\qquad
  K_{rs} = \left(\begin{array}{cccc}
       0 & & & \\
        & \ddots & k & \\
        & k  & \ddots & \\
       & &  & 0 
   \end{array}\right) ,
 \end{aligned}\end{equation}
 i.e. $(E_{rs})_{ab} = \delta_{ar}\delta_{bs} - \delta_{as}\delta_{br}$ etc., for $1\leq r<s\leq n$. In the following we
will consider the complexification of $\mathfrak{sp}(n)$, and denote the imaginary unit in this space by $i$, whereas
the quaternionic number denoted by $i$ so far will not occur explicitly any more. We have the commutation relations
 \begin{align}
  [H_a, Q_b\mp iP_b ]& = (\pm i) 2\delta_{ab} (Q_b \mp iP_b) ,\nonumber\\
  [H_a, E_{rs} \mp i I_{rs}] &= (\pm i)\big(\delta_{ar}-\delta_{as}\big)(E_{rs}\mp i I_{rs}), \\
  [H_a, J_{rs} \mp i K_{rs}] &= (\pm i)\big(\delta_{ar} + \delta_{as}\big)(J_{rs}\mp iK_{rs}).\nonumber
 \end{align}
 and choose the positive roots to be
 \begin{align}
  \alpha_b(H_a) &=  \sqrt 2i \delta_{ab},\qquad\qquad\quad b=1,\dots,n,\nonumber \\
  \beta_{rs}(H_a) &= \frac i{\sqrt 2}\big(\delta_{ar}+\delta_{as}\big) ,\qquad 1\leq r<s\leq n,\\
 \gamma_{rs}(H_a) &= \frac i{\sqrt 2}\big(\delta_{as}-\delta_{ar}\big),\qquad 1\leq r<s\leq n,\nonumber
 \end{align}
 w.r.t. the properly normalized Cartan generator $H'_j= \frac 1{\sqrt 2}H_j$. 

\subparagraph{The coset.}
  Again we choose the embedding of $\mathfrak{sp}(n)$ into the upper left-corner of $\mathfrak{sp}(n+1)$. 
Then the orthogonal torus is generated by $H_{n+1}$, and the
$2n+1$ positive root vectors of $\mathfrak m$ are 
 \begin{align}
  Q_{n+1} - i P_{n+1},\qquad E_{r,n+1} + i I_{r,n+1} ,\qquad J_{r,n+1} - i K_{r,n+1} , 
 \end{align}
 with corresponding positive roots 
\begin{equation}
   R^+\setminus S^+ = \big\{\alpha_{n+1}, \gamma_{r,n+1}, \beta_{r,n+1}\ |\ r=1,\dots,n\big\}.
\end{equation} 
The action of these roots on the Cartan generators $H_a\ (a=1,\dots,n)$ of $\mathfrak k$ is
 \begin{equation}
   \alpha_{n+1}(H_a)=0,\qquad \beta_{r,n+1}(H_a) = -\frac i{\sqrt 2}\delta_{ar},\qquad \gamma_{r,n+1}(H_a) = \frac i{\sqrt 2}\delta_{ar},
 \end{equation} 
 from which we read off that 
 $$    \sum_{\alpha \in R^+\setminus S^+} \alpha\ \Big|_{\mathfrak {sp}(n)} =0 $$
 is again satisfied. There are $n+1 $ invariant spinors, as $\mathfrak{so}(\mathfrak m)
=\mathfrak{so}(4n+1)$, with standard embedding of $\mathfrak{sp}(n)$ into the upper left $\mathfrak{so}(4n)$. They lead
to $\mathcal N=\frac n2 +1$ SUSY for even $n$, and $\mathcal N=\frac {n+1}2$ SUSY for odd $n$. The three-form is somewhat more complicated now, due to the relations
 \begin{equation}
   \beta _{r,n+1} + \gamma_{r,n+1} = \alpha_{n+1}
 \end{equation} 
 between roots in $R^+\setminus S^+$. We get
 \begin{equation}
  \begin{aligned}{}
   H \, &\sim\   2 e^{\alpha_{n+1}} \wedge e^{-\alpha_{n+1}} \wedge h  \\
  & +\sum_{r=1}^n \Big(e^{\beta_{r,n+1}} \wedge e^{-\beta_{r,n+1}} \wedge h + e^{\gamma_{r,n+1}} \wedge e^{-\gamma_{r,n+1}} \wedge h  \Big)  \\
    & -\sqrt 2i \sum_{r=1}^n \Big( N_{\beta,\gamma} e^{\beta_{r,n+1}} \wedge e^{\gamma_{r,n+1}} \wedge e^{-\alpha_{n+1}}   \\
& \qquad \qquad+ N_{-\beta,-\gamma} e^{-\beta_{r,n+1}} \wedge e^{-\gamma_{r,n+1}} \wedge e^{\alpha_{n+1}}\Big), 
  \end{aligned}
 \end{equation} 
 where the constants $N_{\alpha,\beta}$ are defined by $[E_\alpha,E_\beta] = N_{\alpha,\beta}E_{\alpha+\beta}$ for $E_\alpha$ the root vector corresponding to $\alpha$.

\paragraph{Sp($\mathbf{n}$)/SU($\mathbf n$).}
 We use the notation of the last example. Here the roots $\gamma_{rs}$, corresponding to matrices with entries $1,i\in \mathbb H$, belong to $\mathfrak{su}(n)$, whereas 
the other ones with entries $j$ and $k$ belong to $\mathfrak m$. The Cartan basis of $\mathfrak{su}(n)$ is spanned by the elements $\lambda_i = H_{i+1}-H_i$ for $i=1,\dots,n-1$, and there is an orthogonal torus generated by $\sum_i H_i$.
Adding up the roots of $\mathfrak m$, we get
 \begin{equation}
  \sum_{a=1}^n \alpha_a + \sum_{1\leq a<b\leq n} \beta_{ab} = i\frac {n+1}{\sqrt 2}\big(1,\dots,1\big).
 \end{equation} 
 in the basis $H_1,\dots, H_n$, and this vanishes on the Cartan algebra of $\mathfrak {su}(n)$, so that our condition is satisfied. The amount of
supersymmetry could be larger than $\mathcal N=1$; we have $\dim \mathfrak m= n^2+n+1$, and $\mathcal N$ depends
on the number of invariant spinors for the corresponding embedding $\mathfrak{su}(n)\subset
\mathfrak{spin}(n^2+n)$. There are no relations involving only three roots of $\mathfrak m$, so that the three-form becomes
 \begin{equation}
   H\sim \sum_{r<s} e^{\beta_{rs}}\wedge e^{-\beta_{rs}}\wedge h +  \sum_{r} e^{\alpha_{r}}\wedge e^{-\alpha_{r}}\wedge h .
 \end{equation} 

\paragraph{SO($\mathbf{2n}$)/SU($\mathbf n$).}
 Again the rank difference one is one. We define $\mathfrak {su}(n)$ as the subalgebra
of $\mathfrak{so}(2n)$ leaving the subspace $S^0\subset S$ in the spin representation of
$\mathfrak{so}(2n)$ invariant. Upon introduction of the canonical complex structure on $\mathbb R^{2n}$ the subalgebra
 \begin{equation}
   \mathfrak{su}(n) = \{ X\in \mathfrak{so}(2n)\ |\ dS(X)\cdot S^0=0\}
 \end{equation} 
 is determined by the equations
 \begin{equation}\label{Ex:Spin/SU:defEq1}
   X_{12} +X_{34} +\dots +X_{2n-1,2n}=0
 \end{equation} 
 (where $X_{ab}=\delta_{ac}{X^c}_b$) and 
  \begin{align}\label{Ex:Spin/SU:defEq2}
    X_{2a-1,2b-1} = X_{2a,2b},\qquad X_{2a-1,2b}= - X_{2a,2b-1},
  \end{align}
 for $1\leq a <b\leq n$. Denote by $E_{rs}$ ($1\leq r<s \leq 2n$) the standard basis of $\mathfrak{so}(2n)$, with
$(E_{rs})_{ab}= \delta_{ra}\delta_{sb} -\delta_{sa}\delta_{rb}$. 
 A Cartan basis of $\mathfrak{so}(2n)$ is given by $\lambda_a = E_{2a-1,2a}$ for $a=1,\dots,n$, or
 \begin{equation}\label{Ex:Spin/SU:Cartan1}
\lambda_1=\left(\begin{smallmatrix}
 	     0 && 1 && &  &&\\
             -1&&  0 &&  & &&\\
             &   &&& && &\\ 
             & && &&&  &
   \end{smallmatrix}\right), \quad \dots,\quad
 \lambda_n=\left(\begin{smallmatrix}
             &&  &&  && &\\
             &&  &&  && &\\
             &&   && && &\\ 
             & &&& &   0& &   1 \\
             & &&& &  -1& & 0
   \end{smallmatrix}\right).
 \end{equation} 
 A basis of $\mathfrak{su}(n)$ is given by the Cartan generators
 \begin{equation}
   H_a = \lambda_{2a-1} - \lambda_{2a},\qquad a=1,\dots, n-1,
 \end{equation} 
 and combinations of the type $(1\leq a<b\leq n)$
 \begin{equation}
   A^\pm_{ab}=E_{2a-1,2b-1} \pm i E_{2a,2b-1} \mp i(E_{2a-1,2b} \pm i E_{2a,2b}),
 \end{equation} 
 where all signs are correlated. They satsify
  \begin{equation}
   [\lambda_c, A_{ab}^\pm] = \pm i (\delta_{ac} -\delta_{bc})A^\pm_{ab}.
  \end{equation} 
 The Cartan basis $H_1,\dots,H_{n-1}$ can be extended to one of $\mathfrak{so}(2n)$ by adding $H_n= \sqrt{\frac
2n}\sum_{a=1}^n \lambda_a$, which gives rise to an orthogonal torus again.
 A basis of $\mathfrak m$ is given by $H_n$ plus
 \begin{equation}
  B^\pm_{ab} = E_{2a-1,2b-1} \pm i E_{2a,2b-1} \pm i(E_{2a-1,2b} \pm i E_{2a,2b}).
 \end{equation} 
 They have the commutation relations
 \begin{equation}
     [\lambda_c, B_{ab}^\pm] = \pm i (\delta_{ac} +\delta_{bc})B^\pm_{ab},
 \end{equation} 
 so that the roots $\beta_{ab}\in R^+\setminus S^+$, corresponding to $B^+_{ab}$ in the
$(\lambda_1,\dots,\lambda_n$)-basis are
\begin{equation}
  \beta_{ab} = i\big(0,\dots,0,\underset{(a)}1,0,\dots,0,\underset{(b)}1,0,\dots,0\big).
\end{equation}
 In the basis $H_1,\dots,H_{n-1}$ they satisfy the following rules.
\renewcommand{\arraystretch}{1.2}
\begin{table}[H]
 \begin{center}
  \begin{tabular}{|c|c|c|c|} \hline
    $a,b\notin \{2c-1,2c\}$, or & $a=2c-1$, $b=2c$ & $\beta_{ab}(H_c)=0$ &  \\ \hline
    $b=2c-1$, or & $a=2c-1$, $b\neq 2c$ & $ \beta_{ab}(H_c)=+ i$ & $n-2$ times\\ \hline
    $a=2c$, or &$ a\neq 2c-1,\ b=2c$ & $ \beta_{ab}(H_c)=- i$ & $n-2$ times \\ \hline
   \end{tabular}
   \caption{Action of the roots $\beta_{ab}$ ($a<b)$ of $\mathfrak m$ on the Cartan generators of
$\mathfrak{su}(n)$. The number in the last column tells how often the corresponding case occurs, being irrelevant in the
first case.}
  \end{center}
 \end{table}
\renewcommand{\arraystretch}{1}
 We conclude that 
 $$    \sum_{\beta \in R^+\setminus S^+} \beta\ \Big|_{\mathfrak {su}(n)} =0. $$
 Here $\dim \mathfrak m=n^2-n +1$, and again we might have $\mathcal N>1$ SUSY, with $\mathcal N$ equal to the number
of invariant spinors for $\mathfrak{su}(n)\subset \mathfrak{spin}(n^2-n)$. But for $n>3$ this is certainly not relevant
for heterotic string theory, and for $n\leq 3$ we have the exceptional isomorphisms
$\mathfrak{so}(6)=\mathfrak{su}(4)$, $\mathfrak{so}(4) = \mathfrak{su}(2)\oplus\mathfrak{su}(2)$, so that we do not get any essentially new models. The three-form is simply
 \begin{equation}
  H\sim \sum_{a<b} e^{\beta_{ab}} \wedge e^{-\beta_{ab}} \wedge h. 
 \end{equation} 

\paragraph{SO($\mathbf{n+1)}$/SO($\mathbf{n-1}$).} 
 We use the notation of the previous example.
 Contrary to $\mathfrak{so}(n)/\mathfrak{so}(n-1)$ this space admits invariant spinors and 
 the common root space decomposition. Suppose first that $n$ is odd, so that we consider $\mathfrak{so}(2k)/\mathfrak{so}(2k-2)$. Denote the common Cartan generators by $\lambda_1,\dots,\lambda_{k-1}$, and the remaining one in $\mathfrak m$ by $\lambda_k$. The positive roots of $\mathfrak{so}(2k)$ in the $\lambda$-basis are
 \begin{equation}
  \alpha_{ab} = i(\dots,\underset{(a)} {-1},\dots,\underset{(b)}{1},\dots),\qquad   \beta_{ab} = i(\dots,\underset{(a)} 1,\dots,\underset{(b)}{1},\dots),
 \end{equation}  
 for $1\leq a<b \leq k$, and the ones belonging to $\mathfrak m$ are
 \begin{equation}
  R^+\setminus S^+ = \{ \alpha_{ak},\beta_{ak}\ |\ a=1,\dots,k-1\}.
 \end{equation} 
 From 
 \begin{equation}
   \beta_{ak}(\lambda_j) = -\alpha_{ak}(\lambda_j) =i,\qquad \forall j<k,
 \end{equation} 
 it follows immediately that the condition $\sum_{\alpha\in R^+\setminus S^+ }\alpha \big|_{\mathfrak {so}(2k-2)}=0$ is 
satisfied. What is the amount of supersymmetry preserved? We have that $\mathfrak m$ is twice the fundamental
representation of $\mathfrak{so}(2k-2)$, plus the trivial from the additional Cartan generator. The
complex spin representation has the nice feature that for two even-dimensional representations $V,W$ we have
  \begin{equation}
    S(V\oplus W) = S(V)\otimes S(W),
  \end{equation} 
 so that, ignoring the additional Cartan generator which we have taken into accout already, $S(\mathfrak m)$ is simply
the tensor product of the spin representation with itself. This contains the trivial representation only twice
\cite{Var04}, upon introduction of a $\mathfrak{spin}(2n)$-invariant inner product and an orthonormal basis $\{\psi_i\}$
of $S$ the invariant elements are given by 
 \begin{equation}
   \sum_i \psi_i \otimes \psi_i,\qquad \text{and}\qquad \sum_i \psi_i \otimes \Gamma\psi_i,
 \end{equation}
 with $\Gamma$ the chirality element.
 From the embedding $\mathfrak{so}(2k-2)\subset \mathfrak{su}(\mathfrak m)$ we knew there would be at least 2 spinors
invariant under the complexified Lie algebra, and as there are no further ones we get indeed $\mathcal N=1$ SUSY. 
 The situation with $\mathfrak{so}(2k+1)/\mathfrak{so}(2k-1)$ is very similar, the only additional positive roots appearing
are the
 \begin{equation}
   \gamma_a = (0,\dots,0,\underset{(a)} {+i},0,\dots,0 ),\qquad 1\leq a\leq k,
 \end{equation} 
 with $\gamma_k$ the only new root in $R^+\setminus S^+$. But this one vanishes on all of the $\lambda_j$ for $j<k$, so
we have the condition satisfied here as well. The dimension of SO($n+1$)/SO($n-1)$ is $2n-1$, together with the
additional $\mathbb R$ factor supporting the linear dilaton this gives dimension $2n$, and the only essentially new
model with dimension low-enough to be of interest for heterotic string theory is the seven-dimensional SO(5)/SO(3).
Below we will consider another embedding of SO(3) into SO(5) which also preserves $\mathcal N=1$. Here $H$ is
 \begin{equation}
  H\sim \sum_{a=1}^{k-1} \Big(  e^{\alpha_{ak}}\wedge e^{-\alpha_{ak}}\wedge h +  e^{\beta_{ak}}\wedge e^{-\beta_{ak}}\wedge h \Big) 
 \end{equation} 
 for $\mathfrak{so}(2k)/\mathfrak{so}(2k-2)$, whereas on $\mathfrak{so}(2k+1)/\mathfrak{so}(2k-1)$ it has additionally the term $e^{\gamma_{k}}\wedge e^{-\gamma_{k}}\wedge h$.

\paragraph{Spin(7)/$\mathbf{G_2=S^7}$.}
 For this example we get a solution as well, but our construction is not applicable, 
the concrete calculations below mainly serve as a preparation for the next example,
$G_2/$SU(3). Here we can use the following argument instead. 
We have dim$(\mathfrak m)=7$, and $\mathfrak m$ is not a sum of trivial representations of $\mathfrak g_2$. But then
it has to be the irreducible 7-dimensional one, as this is the lowest possible dimension for irreducible $\mathfrak g_2$
representations.
 The spin representation of $\mathfrak g_2$ over $\mathfrak m$ then comes from the standard embedding into
$\mathfrak{spin}(7)$, which will be given below, and it is well-known that this leaves invariant exactly one
Majorana spinor $\psi$. On the other hand $\gamma(H)$ commutes with $\mathfrak g_2$ acting on spinors, so that $\psi$
must be an eigenspinor:
 \begin{equation}
  \gamma(H)\psi =i \lambda \psi,\qquad \lambda \in \mathbb R.
 \end{equation} 
  For a 
properly normalized linear function $\phi$ on $\mathbb R$, we then get $\gamma(d\phi)\sim$ id on the 1-dimensional
spinor space over $\mathbb R$, and from a suitable non-zero 1D spinor $\kappa$ we can construct an 8-dimensional spinor
$\kappa\otimes \psi$ such that
 \begin{equation}
  \gamma\Big(d\phi - \frac 1{12} H\Big) \kappa\otimes \psi=0.
 \end{equation}  
 This solution preserves $\mathcal N=1$ SUSY in the orthogonal 2-dimensional space, and the same argument applies to
other 
seven-dimensional spaces with nearly parallel $G_2$-structure and one Killing spinor. We will work out one example in
more detail below, for the space SO(5)/SO(3) with a non-standard embedding.  Other examples of this type are given by
the Aloff-Wallach spaces $N(k,l)$.

\bigskip

 In the remainder of the paragraph we work out the embedding of $\mathfrak g_2$ into $\mathfrak {spin}(7)$, and its
root space decomposition.
We define $\mathfrak g_2$ as the subalgebra of $\mathfrak{spin}(7)$ fixing a given Majorana spinor
\cite{FrKMS97}. This condition amounts to seven equations, which we write in terms of the antisymmetric
$X_{ab}:=\delta_{ac}{X^c}_b$, $X\in \mathfrak{so}(7)$:
 \begin{align}\label{g_2DefEqns}
  X_{17} &= - X_{36} - X_{45},\qquad X_{27} =X_{46}- X_{35}, \nonumber\\
  X_{37} &=  X_{25} + X_{16},\qquad \quad X_{47} =X_{15}- X_{26} ,\nonumber\\
  X_{57} &= - X_{14} - X_{23},\qquad X_{67} =X_{24}- X_{13} ,\\
  &\qquad X_{12} + X_{34} + X_{56} =0. \nonumber
 \end{align}
 Denote by $E_{ij}$ the 7$\times 7$-matrix with $(E_{ij})_{ab} = \delta_{ai}\delta_{bj} - \delta_{bi}\delta_{aj}$. A
Cartan basis for $\mathfrak {so}(7)$ is given by
 \begin{align}\label{so(7)Cartanbasis}
     \lambda_1 &= E_{12} = \left(\begin{array}{ccc}
         0 & 1 & \\
        -1 & 0 & \\
           & & 0_{5\times 5}
   \end{array}\right), \nonumber\\
     \lambda_2 &= E_{34} = \left(\begin{array}{cccc}
         0_{2\times 2} &  & & \\
          & 0 & 1 & \\
           & -1& 0 & \\
          & & & 0_{3\times 3}
  \end{array}\right), \\
     \lambda_3 &= E_{56} = \left(\begin{array}{cccc}
        0_{4\times 4} & & & \\
         & 0 & 1 &  \\
        & -1 & 0 &  \\
        & & & 0
   \end{array}\right). \nonumber
 \end{align}
 and we have 
 $$[\lambda_1 , E_{1j} \pm i E_{2_j} ] = \pm i (E_{1j} \pm iE_{2j}),\qquad j>2,$$
 etc. We can choose a basis of $\mathfrak g_2$ with Cartan generators 
 \begin{equation}
  \sqrt 6 H_1=  2\lambda_1 -\lambda_2-\lambda_3,\qquad\sqrt 2 H_2 =\lambda_2 -\lambda_3,
 \end{equation} 
 and the other basis elements
 \begin{align}
  E_{36}-E_{45} ,&\qquad 2E_{17} - E_{36}-E_{45}, \\
  E_{46}+E_{35}, &\qquad 2E_{27} +E_{46}-E_{35}, \nonumber
 \end{align}
 etc. (cf \eqref{g_2DefEqns}; the elements listed here are deduced from the first line in \eqref{g_2DefEqns}, the second
and third line give rise to another 8 analogous basis elements, so that we end up with 14 generators). Now we introduce
the following complex linear combinations of these:
\begin{equation} 
\begin{aligned}{}
   A_1& = E_{36}-E_{45} + i(E_{46}+E_{35}) \\
   A_2 &= E_{25}-E_{16}-i(E_{15}+E_{ 26}) \\
   A_3 &= E_{14}-E_{23} -i(E_{24}+E_{13}) \\
  B_1 &= 2(E_{17} + iE_{27}) -E_{36}-E_{45} + i(E_{46}-E_{35}) \\
  B_2 &= 2(E_{37}+iE_{47}) + E_{25}+E_{16} +i(E_{15}-E_{26}) \\
 B_3 &= 2(E_{57}-iE_{67}) -E_{14}-E_{23}+i(E_{13}-E_{24}) .
 \end{aligned}
\end{equation}
 The positive roots w.r.t. $H_1,H_2$ are $i/\sqrt 2$ times:
 \begin{align}
   \alpha_1=(0,2),\qquad\quad\alpha_2= (\sqrt 3,1)&,\qquad\quad\alpha_3= (-\sqrt 3,1), \\
   \beta_1=(2/\sqrt 3,0),\qquad \beta_2=(-1/\sqrt 3,&1),\qquad \beta_3=(1/\sqrt 3,1),\nonumber
 \end{align}
 where the first line has the $A$-roots, and the second one the $B$-roots. The ordering of the roots is
 $$ \alpha_1> \alpha_2 > \beta_3>\beta_2 >\alpha_3 >\beta_1.$$
 In the orthogonal complement of $\mathfrak g_2$ in $\mathfrak{so}(7)$ let us introduce the third Cartan generator
 \begin{equation}
  H_3 = \sqrt {\frac 13} (\lambda_1+\lambda_2+\lambda_3) . 
 \end{equation} 
 Now it turns out that $[H_3,A]=0$, but $[H_3,B]\neq 0$. Thus the root space decomposition of $\mathfrak g_2$ cannot be
lifted to one of $\mathfrak{spin}(7)$, and our construction does not apply to this coset, but the concrete
realization of $\mathfrak g_2$ will still be useful for the following example.

\paragraph{$\mathbf{G_2/}$SU(3) $\mathbf{=S^6}$.}

 If we put $X_{a7}=0$ in \eqref{g_2DefEqns}, then we obtain the defining equations of $\mathfrak{su}(3)$ inside
$\mathfrak{so}(6)$, and thereby realize $\mathfrak{su}(3)$ as a subalgebra of $\mathfrak g_2$. A basis of $\mathfrak
{su}(3)$ is given by the Cartan generators $H_1,H_2$, and $A_1,A_2,A_3$ plus their complex conjugates. Therefore the
roots $\alpha_1,\alpha_2,\alpha_3$ are in $S^+$, whereas $\beta_{1,2,3,}$ are in $R^+\setminus S^+$. The roots of
$\mathfrak g_2$ satisfy the following relations:
 \begin{align}\label{g_2rootrelations}
  \alpha_1 = \alpha_2 + \alpha_3,\qquad \alpha_1&=\beta_2+\beta_3 ,\qquad \alpha_2 = \beta_1+\beta_3 ,\\
  \beta_2=\alpha_3+\beta_1 ,&\qquad \beta_3=\beta_1+\beta_2,\nonumber
 \end{align}
 contradicting our assumption that roots in $R^+\setminus S^+$ are higher than those of $S^+$. In this
case this is not a matter of bad choice of positive roots, it is simply not
possible to satisfy the condition. However, if we now choose the complex structure on $\mathfrak m$ as follows, by
giving its $+i$ eigenspace ((1,0)-vectors) and $-i$ eigenspace ((0,1)-vectors):
 \begin{equation}
   +i:\quad B_1,B_2,\overline B_3, \qquad -i: \quad \overline B_1,\overline B_2,B_3,
 \end{equation} 
 then \eqref{g_2rootrelations} shows that we have again $\mathfrak k=\mathfrak{su}(3)\subset \mathfrak u(\mathfrak m)$.
 $J$ is also compatible with the metric, but leads to a three-form $H\in \Omega^{(3,0)}\oplus
\Omega^{(0,3)}$, instead of the usual $H\in \Omega^{(2,1)}\oplus \Omega^{(1,2)}$. Now it is not the positive roots in
$R\setminus S$ corresponding to creation and annihilation operators, but the roots $\beta_1,\beta_2$ and $-\beta_3$.
Therefore also the spin representation of the Cartan generators $H_1,H_2\in \mathfrak {su}(3)$ changes slightly (the
$\zeta^3$ part changes sign) and is now given by
 \begin{equation}
  \begin{aligned}{}
   \widetilde{\text{ad}}(H_1) &= \frac i{2\sqrt 6} \Big(2(\overline \zeta^1 \zeta^1 -\zeta^1 \overline \zeta^1 )- \overline \zeta^2 \zeta^2
+\zeta^2 \overline \zeta^2 -\overline \zeta^3 \zeta^3 +\zeta^3 \overline \zeta^3 \Big),\\
  \widetilde{\text{ad}}(H_2)&= \frac i{2\sqrt 2} \Big(\overline \zeta^2 \zeta^2 -\zeta^2 \overline \zeta^2 - \overline \zeta^3 \zeta^3
+\zeta^3 \overline \zeta^3\Big),
  \end{aligned}
 \end{equation} 
  This leaves invariant $S^0\oplus S^m$, being equivalent to $\beta_1+\beta_2 -\beta_3=0$.
 The three form is (independently of the complex structure) 
 \begin{equation}
  H =-8i( e^{12\overline 3} - e^{\overline{12}3}) \ \in \ \Omega^{(3,0)} \oplus \Omega^{(0,3)}.
 \end{equation} 
 It acts on the spinors through
 \begin{equation}
  \gamma(H) = -24 i \big(\zeta^1\zeta^2\zeta^3 - \overline \zeta^1\overline \zeta^2\overline \zeta^3\big),
 \end{equation} 
 and in particular on $\mathcal S^0\oplus S^3$ as
 \begin{equation}
   \gamma(H)\cdot 1 = -24 i w^{123} ,\qquad \gamma(H) w^{123} = -24i\cdot 1.
 \end{equation} 
 Here we do not have additional flat directions at our disposal to define the linear dilaton on, as the rank of
$\mathfrak k$ is equal to rk$(\mathfrak g_2)$. Now we could try to introduce one by hand and consider $\mathbb R \oplus  \mathfrak g_2/\mathfrak{su}(3)$ with $\phi$ linear on $\mathbb R$. It turns out that this does not lead to any solutions, which is probably due to the fact that the tensor product of the spinor spaces in dimension three and seven does not give the ten-dimensional spinor space, which in turn follows from dimensional reasons. Let us then try adding $\mathbb R^2$. In this case we can put $\phi(x_1,x_2) = 2x_2$, introduce a fourth complex spinor variable $w^4\simeq \frac 12(dx^1 +idx^2)$, and obtain $\gamma(d\phi) = 2i(\overline \zeta ^4-\zeta^4)$. Of the 4 invariant Majorana-Weyl spinors
 \begin{equation}
  \begin{aligned}{}
    \eta & = 1+w^{1234},\qquad\quad \mu = w^4 - w^{123}  \\
    \kappa &= i(1-w^{1234}) ,\qquad \lambda =i(w^4 +w^{123}) 
  \end{aligned}
 \end{equation} 
 the first two are annihilated by $\gamma(d\phi-\frac 1{12}H)$, the other two are not. Thus we get two solutions on $\mathbb R^2 \times G_2/$SU(3). As the dilaton depends only on one of the two additional directions, we should view the other $\mathbb R$ as belonging to the orthogonal space $\mathbb R^{2,1}$, with $\mathcal N=1$ supersymmetry preserved, noting that a three-dimensional Majorana spinor can be combined with $\eta$ and $\mu$ in a unique way into a 10-dimensional Majorana-Weyl spinor. 

\bigskip 

 The almost complex structure on $G_2/$SU(3) we defined here makes it
into a so-called nearly K\"ahler manifold. In 6 dimensions only 4 manifolds of this type are known, and they are all
cosets \cite{Butruille}:
  \begin{align*}
    SU(3)/U(1)\times U(1),\qquad & \quad Sp(2)/Sp(1)\times U(1),\\
   G_2/SU(3)=S^6,\quad\qquad & SU(2)^3/SU(2)_{\text{diag}}=S^3\times S^3.
  \end{align*}
 The same construction we presented here for $G_2/$SU(3) can be applied to any of them, but for
SU(2)$^3$/SU(2)$_{\text{diag}}$ it has to be modified
slightly, with the two Cartan generators of $\mathfrak m$ combined into one complex vector,
effectively treating SU(2)$^3$ as a rank 1 group. Except for G$_2$/SU(3) the nearly K\"ahler cosets have
higher-dimensional generalizations, and one may wonder whether the construction applies to these.\\

From the well-known fact that Sp($n)\times $U(1) as a subset of Spin($4n$) has no invariant spinors (which is equivalent
to the fact that quaternion-K\"ahler manifolds have no parallel spinors), it can be deduced that the equal rank coset
Sp$(n)/$Sp$(n-1)\times $U(1) has no invariant spinors for $n>2$. 
 Let us then consider SU(2)$^n$/SU(2)$_\text{diag}$, where we take $n$ odd for simplicity. $\mathfrak m$ consists of
$n-1$ copies of the adjoint representation, with weights $n-1$ times (2),(0),(-2), in Dynkin label notation.
 From this one can determine the weights of the spin representation, and deduce that the trivial representation occurs
with multiplicity
 \begin{equation}
   \frac 12(n-1)\bigg[ \binom{n-1}{\frac {n-1}2 }-\binom{n-1}{\frac {n-3}2 }\bigg] =  \frac
{n-1}{2n}\binom{n}{\frac{n-1}2}.
 \end{equation} 
 This case thus generalizes, but with a lot of invariant spinors. We leave it open whether or not the dilatino equation
can be solved in general, but see the discussion at the end of the following paragraph.

\paragraph{SU($\mathbf{n+1}$)/U(1)$^{\mathbf n}$} is quite interesting in itself, as the problem of finding invariant
spinors becomes a simple combinatorial problem on the roots, so we treat it in detail. Furthermore, for odd
$n$ there are no invariant spinors, contrary to all the other non-symmetric examples considered so far. The subalgebra
is a Cartan algebra here, and the
decomposition of $\mathfrak m\otimes \mathbb C$ into
irreducibles is given by the root space decomposition. The positive roots are
 \begin{equation}
   \alpha_{ab} = i(0,\dots,0,\underset{(a-1)}{-1},\underset{(a)}{+1},0,\dots,0,\underset{(b-1)}{+1},\underset
{(b)}{-1},0,\dots,0),
 \end{equation} 
 with $1\leq a < b \leq n+1$, and terms that run out of the $n$ entries of the root vector are dropped. They satisfy
the relations
 \begin{equation}\label{SUnRels}
   \alpha_{ab} + \alpha_{bc} = \alpha_{ac}.
 \end{equation} 
 Similarly, the weights of the spin representation determine the irreducible representations in $S(\mathfrak m)$, and
there are invariant spinors if and only if the zero weight occurs. As usual, the weights of $S(\mathfrak m)$ are of the
form $\frac 12\sum_{\alpha \in R^+} \pm \alpha$, with arbitrary combinations of signs. 
 \begin{prop}
   For even $n$ invariant spinors exist, but not for odd $n$.
 \end{prop}
\begin{proof} $n$ even:
  We have to show that there exists a certain combination of signs such that $\sum_{ab} \pm \alpha_{ab}=0$. Group the
roots into pairs of 3 or 4 elements according to the following triangle:
 
  \renewcommand{\arraystretch}{1.4}
   \begin{table}[h]\centering
  \begin{tabular}{cccccc}       
   &&$\vdots $&&&  \\ \hline 
  \multicolumn{1}{|c}{$\alpha_{17}$} &    \multicolumn{1}{c|}{$\alpha_{27}$} &   \multicolumn{1}{|c}{$\alpha_{37}$}  &  
\multicolumn{1}{c|}{$\alpha_{47}$}  &   \multicolumn{1}{|c}{$\alpha_{57}$}  &   \multicolumn{1}{c|}{$\alpha_{67}$} \\
  \multicolumn{1}{|c}{$\alpha_{16}$} &    \multicolumn{1}{c|}{$\alpha_{26}$} &   \multicolumn{1}{|c}{$\alpha_{36}$}  &  
\multicolumn{1}{c|}{$\alpha_{46}$}  &   \multicolumn{1}{|c}{$\alpha_{56}$}  &  \multicolumn{1}{c|}{}  \\ \hline
   \multicolumn{1}{|c}{$\alpha_{15}$} &    \multicolumn{1}{c|}{$\alpha_{25}$} &   \multicolumn{1}{|c}{$\alpha_{35}$}  &
\multicolumn{1}{c|}{$\alpha_{45}$}  &    &   \\ 
  \multicolumn{1}{|c}{$\alpha_{14}$} &    \multicolumn{1}{c|}{$\alpha_{24}$} &   \multicolumn{1}{|c}{$\alpha_{34}$}  & 
\multicolumn{1}{c|}{} &     &   \\ \cline{1-4}
  \multicolumn{1}{|c}{$\alpha_{13}$} &    \multicolumn{1}{c|}{$\alpha_{23}$} &   & &     &   \\ 
  \multicolumn{1}{|c}{$\alpha_{12}$} &     \multicolumn{1}{c|}{} &    & &     &   \\ \cline{1-2}
    \end{tabular}
 \caption{Grouping the roots of $\mathfrak{su}(n+1)$ into blocks of 3 or 4 elements}
 \end{table}
\renewcommand{\arraystretch}{1} 
 Inside each box there is a relation of the type 
 \begin{equation}
  \alpha_{12} + \alpha_{23} -\alpha_{13} =0,\qquad \alpha_{14}-\alpha_{24} -\alpha_{15} + \alpha_{25} =0,\quad \dots.
 \end{equation} 
 Fix the signs of the roots $\alpha_{ab}$ accordingly and denote them by sign$(a,b)$, then we certainly have 
 $$ \sum_{1\leq a<b\leq n+1} \text{sign}(a,b) \alpha_{a,b} =0,$$
 and the corresponding spinor is invariant. On the other hand this shows clearly that there will be many more invariant
spinors for large $n$, but we will not try to determine their number. 
\bigskip
 
 $n$ odd: The reason that we could sum all the roots to zero by an appropriate choice of signs in the even $n$ case was
that in every box the indices that appeared, appeared exactly twice, as in the relations \eqref{SUnRels}. Adding
another row to the diagram of the form
 \begin{equation}
   \alpha_{1,n+1} \qquad \alpha_{2 ,n+1} \qquad \dots\qquad \alpha_{n,n+1}
 \end{equation} 
 for odd $n$ we get every index appearing with odd multiplicity. It follows that the roots cannot add up to zero anymore.
\end{proof}

 In the case of even $n>2$ there are several invariant Majorana spinors and $\gamma(H)$ acts non-trivially
on them, so that it is far from obvious whether this action can be compensated for by a linear dilaton on an
extra $\mathbb R$ factor, and for a well-chosen spinor $\varepsilon$. Analyzing our explicit construction of the
spinors, this seems rather unlikely, although a proof that it is not possible probably requires more elaborate
techniques. The
models with $n\geq 3$ are not relevant for heterotic supergravity anyway, but they would very much break the pattern of 
the other solutions: they are of even dimension greater than six, so that their cone cannot support parallel spinors
\cite{Wang89}, contrary to all the solutions we found so far (cf. the discussion in the conclusion). We therefore
conjecture that it is not possible to solve the dilatino equation on them. The same applies to SU(2)$^n/$SU(2) for odd
$n$, and every other naturally-reductive coset of even dimension different from six. 

\paragraph{SO(5)/SO(3).}
 As an example where the common root space decomposition does not exist, and the subalgebra is not in $\mathfrak
u(\mathfrak m)$, we consider the following embedding of $\mathfrak{so}(3)$ into $\mathfrak{so}(5)$, whose coset manifold
is known to possess a nearly parallel $G_2$-structure \cite{Bryant87, FrKMS97}:
\begin{equation}
 \left(\begin{array}{ccc}
                                                     0& \alpha &\beta \\ 
                                                     -\alpha & 0& \gamma \\
                                                      -\beta & -\gamma &0 \\
                                                   \end{array}\right)
 \mapsto  \left(\begin{array}{ccccc}
                 0& \alpha &-\gamma &\beta & \sqrt 3 \gamma \\ 
              -\alpha & 0& -\beta &-\gamma &-\sqrt 3\beta\\
                \gamma & \beta & 0 & 2\alpha & 0 \\
	        -\beta & \gamma & -2\alpha & 0&0 \\
	        -\sqrt 3 \gamma & \sqrt 3\beta & 0 &0&0 \\
            \end{array}\right)
\end{equation} 
 In terms of our basis vectors 
\begin{equation}
 \begin{aligned}{}
 \lambda_1 &= E_{12},\qquad \lambda_2 = E_{34}, \\
  A &= E_{13} + E_{24} +i(E_{23}-E_{14} ),\\
  B &=  E_{13} - E_{24} +i(E_{23}+E_{14} ), \\
  C_1 &= E_{15} +iE_{25}, \\
  C_2 &= E_{35} + i E_{45}
 \end{aligned}
\end{equation} 
 of $\mathfrak {so}(5)$ defined above, the subalgebra and its complement are spanned by
 \begin{equation}
  \begin{aligned}{}
    \mathfrak{so}(3) &:\qquad H_1:=\lambda_1 +2\lambda_2,\quad D:= A-\sqrt 3 \overline C_1,\quad
\overline
D\\
   \mathfrak m &: \qquad H_2:= 2\lambda_1 -\lambda_2 ,\quad B,\overline B,C_2,\overline C_2, \quad E:= \overline A+
\frac 2{\sqrt 3}C_1,\ \ \overline E.
  \end{aligned}
 \end{equation} 
 This is not a common root space decomposition, as $[H_2,D]\neq 0$, and $[H_2, E]\notin \mathfrak m$. It is a
naturally reductive coset though. There is no complex structure on $\mathfrak m$ such that $\mathfrak{so}(3)\subset
\mathfrak u(\mathfrak m)$, and it will be important to know precisely the action of $\mathfrak k$ on $\mathfrak m$.
After a rescaling 
 $$ B\mapsto \sqrt{\frac52}B,\quad C_2\mapsto \sqrt 5C_2,\quad E\mapsto \sqrt{\frac 32}E,$$
  it is given by
 \begin{equation}
  \begin{aligned}{}
    [H_1,B] &= 3iB ,\qquad\quad [H_1,C_2] = 2iC_2,\qquad [H_1, E] = iE ,\\ 
    [D,B] &= -\sqrt 6 C_2 ,\quad \ [D,C_2] =\sqrt {10} E,\qquad [D,E]= 2\sqrt 6i H_2, \\
    [D,\overline B] &= 0,\qquad \qquad[D,\overline C_2] = \sqrt 6\overline B,\qquad\ \  [D,\overline E]= -\sqrt {10}
\overline C_2 ,\\
   [H_1,H_2] &=0 ,\qquad \qquad\ [D,H_2]  =-\sqrt 6i\overline  E   .
  \end{aligned}
 \end{equation} 
 In particular the weights of $\mathfrak{so}(3)$ in the Dynkin label notation are (6),(4),(2),(0),(-2),(-4),(-6),
which belong to the irreducible seven-dimensional representation with highest weight (6) (in general the highest-weight
module to highest weight $(n)$ has dimension $n+1$). We write $\mathfrak m = (6)$. For the linear dilaton we will
need to extend this to $\mathfrak m'= \mathfrak m \oplus \mathbb R$, which is the $\mathfrak k$-module $\mathfrak
m'=(6)\oplus (0)$, and the weights of the corresponding spinor module are
 \begin{equation}
   \Omega (S(\mathfrak m')) = 2\times \Big\{(6),(4),(2),2\times (0),(-2),(-4),(-6)\Big\},
 \end{equation} 
 showing that $S(\mathfrak m') = 2\times \big((6)\oplus (0)\big)$. We conclude that there are invariant spinors again,
and in order to identify them we need to introduce a complex structure. Take its $+i$ eigenspace to be spanned by
 \begin{equation}
   +i :\qquad \overline B,C_2,E , H_2 -i \partial_x,
 \end{equation} 
 where $\partial_x $ is the standard basis vector field of $\mathbb R$. Then the commutation relations above tell us
how the basis elements of $\mathfrak k$ are quantized:
\begin{equation}
 \begin{aligned}{}
   \widetilde{\text{ad}}(H_1) &= \frac i2 \Big( 3( \zeta^B\overline\zeta ^B + \overline\zeta^B  \zeta^B) + 2(\overline \zeta^C\zeta ^C +
\zeta^C\overline  \zeta^C) +\overline \zeta^E\zeta ^E + \zeta^E\overline  \zeta^E\Big),\\
  \widetilde{\text{ad}}(D) &=  \sqrt 6\overline \zeta ^C\overline\zeta^B + \sqrt{10}\zeta^C\overline \zeta^E +\sqrt 6 i\zeta^E(\zeta^H
+ \overline \zeta ^H), \\
  \widetilde{\text{ad}}(\overline D) &= \sqrt 6\zeta ^C \zeta^B + \sqrt{10}\overline \zeta^C\zeta^E -\sqrt 6
i\overline \zeta^E(\zeta^H
+ \overline \zeta ^H).
 \end{aligned}
\end{equation} 
 Invariant under all of $\mathfrak{so}(3)$ are
 \begin{equation}
    1 - i w^{BCEH},\qquad   w^H -i w^{BCE},\qquad 
 \end{equation} 
 and these multiplied by $i$, or the Majorana-Weyl spinors
\begin{equation}
 \eta=(1+i) 1 +(1-i) w^{BCEH} ,\qquad \mu=(1-i)w^H - (1+i)w^{BCE}.
\end{equation} 
 It is much more effort to determine the action of $\gamma(H)$ on these spinors than it was in the $\mathfrak{su}(n)$
case. We need the commutation relations between elements of $\mathfrak m$. They are as follows:
\begin{equation}
 \begin{aligned}{}
  [H_2, B] &= iB,\quad [H_2,\overline B,]= -i\overline B,\quad [H_2,C_2]= -i C_2 ,\quad [H_2,\overline C_2]=iC_2, \\
  [H_2,E] &= -\sqrt 6i D-iE,\qquad \qquad\ \ [H_2,\overline E] = \sqrt 6i \overline D+i\overline E,\\
 [B,\overline B]  & = 6i H_1 + 2i H_2 ,\qquad\qquad \quad\ [C_2,\overline C_2] = 4i H_1 - 2iH_2 ,\\
 [E,\overline E] &= 2iH_1-2iH_2, \qquad \qquad \quad \ [B,\overline C_2] =  -\sqrt 6 \overline D +2 E ,\\
 [B,C_2] &= 0, \qquad\ \ [B,E] = 0 ,\qquad\ \  [B,\overline E] =-2C_2, \\
 [C_2,E] &= 2B,\qquad [C_2,\overline E] = \sqrt{10} \overline D,
 \end{aligned}
\end{equation} 
 from which we read off that
 \begin{equation}
  H= 4\big(e^{B\overline{ C E}}+e^{ \overline B CE}\big) -2i\big(e^{\overline BB}+e^{C\overline C} +e^{E \overline
E}\big)\wedge \big(e^H + e^{\overline H}\big),
 \end{equation} 
 leading to the result
 \begin{equation}
   \gamma(H)\eta = -42\mu,\qquad \gamma(H)\mu =42\eta.
 \end{equation} 
 Defining $\phi:\mathbb R\rightarrow \mathbb R$ by $\phi(x)=-\frac 72 x$, we obtain $\gamma(d\phi) = -\frac 72i(\overline \zeta^H -\zeta ^H)$, and
 \begin{equation}
   \begin{aligned}{}
    \gamma\Big(d\phi - \frac 1 {12} H\Big) \eta &= 0,\\
    \gamma\Big(d\phi - \frac 1 {12} H\Big) \mu &= -7\eta \neq 0.
   \end{aligned}
 \end{equation} 
 There is thus exactly one Majorana-Weyl spinor on $\mathbb R\times $SO(5)/SO(3) solving the heterotic supersymmetry
equations, leading to $\mathcal N=1$ SUSY in the orthogonal 2-dimensional space. More abstractly, we could have
argued as in the 
case Spin(7)$/G_2$, using that SO(5)/SO(3) carries a nearly parallel $G_2$-structure with one Killing spinor
\cite{FrKMS97}.

\section{Conclusion}

 We have shown that there is a large class of solutions to the heterotic supergravity BPS-equations \eqref{IntroSUSYEqtns} and the Bianchi identity \eqref{IntroBianchi} on spaces of the form
$\mathbb R^{p,1}\times G/K$, with a non-symmetric coset $G/K$ and homogeneous, nonzero fields $H,d\phi, A$. The dilaton
$\phi$ has to be taken as a linear function on $\mathbb R$, whereas the other fields live on $G/K$. In the special case
of trivial $K$, one obtains a WZW model coupled to a linear dilaton, and the quantization condition on $H$ fixes the
volume modulus to a discrete set of possible values. For
nontrivial $K$, the scale is completely fixed in terms of $\alpha'$, but this time by the Bianchi identity, which was
trivially satisfied in the WZW case. It is not clear however how this relation behaves under inclusion of higher order
$\alpha'$ corrections. \\
 \indent All of the models are therefore free of a volume modulus, one of the
physically pleasing features which stands in sharp
contrast to the existence of the linear dilaton, which rules out these spaces as models
for our observable universe. On the other hand the WZW model for SU(2) with linear dilaton describes the near
horizon geometry of NS5-branes \cite{Callan91, Aharony98}.

\bigskip
 
 The gaugino equation $\gamma(F)\varepsilon=0$ (together with the condition $\nabla^- \varepsilon=0$) can be seen as a generalized instanton equation, and is quite interesting in its own right. Here we considered only the simplest solution $F=R^-$ \eqref{Cosets:canCurvature1}, which was shown to solve the Bianchi identity as well in the discussion preceding \eqref{BianchiSolved}. More solutions are known to exist on spaces of the form $\mathbb R^p\times G/H$ however \cite{HILP09, ILPR09, HP10, BILL10}, and using these one could try to construct more general supergravity solutions.

\bigskip

 Except for the WZW models, the equations of motion at order $\alpha'$ are not all satisfied by our models, and
therefore it is not clear a priori whether a corresponding CFT really exists. Interestingly, the article \cite{Eguchi03}
by 
Eguchi, Sugawara and Yamaguchi has a construction of CFTs on exactly the kind of cosets we considered here, coupled to a
linear dilaton. Furthermore, they come to the same conclusions about preservation of supersymmetry, namely that one
does not get spacetime supersymmetry for SO($n+1)$/SO($n)$, whereas all the non-symmetric cosets are supersymmetric. But
a complete disagreement exists on the level; like in WZW models they find a discrete set of allowed levels, whereas in
our case it seems to be fixed. Additionally it has been argued that the gauged WZW-CFTs used in \cite{Eguchi03} do
not describe sigma models with target space the usual geometric cosets, but that the latter require a
modified gauging, which has successfully been applied to quotients with respect to abelian subgroups only \cite{Is04}.

\bigskip

 There is an interesting connection to the cone construction, which was emphasized in \cite{Eguchi03}. 
Consider a compact manifold $M$ of Sasaki-Einstein, 3-Sasakian, nearly K\"ahler (in 6D), or nearly parallel $G_2$ type
(in 7D). Then its cone $\mathbb R_{>0} \times M$, equipped with a particular cone metric, has a parallel spinor and
holonomy contained in SU($n$), Sp($n)$, $G_2$ or Spin(7) respectively \cite{Baer93, Boyer07}. Thus the cone can be used to
construct BPS-vacua of string theory (any type, and also of $M$ theory) of the form
 \begin{equation}
   \mathbb R^{p,1} \times \mathbb R_{>0} \times M,
 \end{equation} 
 with $H=d\phi=0$, and $A$ equal to the Levi-Civita connection, for the heterotic case. The amount of supersymmetry preserved is the same as in our models with the radial direction replaced by the linear dilaton. 
 Interest in these models is
largely due to the conjectured duality between the worldvolume conformal field theory on branes placed at the tip of the
cone and certain string theories on anti-de Sitter spaces, which describe the near-horizon geometry of the brane 
\cite{Klebanov-Witten, Acharya98}. For a proposal for the CFT side see \cite{Martelli09}. 

\bigskip

 All of the cosets considered in this paper carry indeed one of the non-integrable structures listed above, although in general with a metric different from the one we used. On SU($n+1)$/SU($n)=S^{2n+1}$, Sp$(n)$/SU($n)$,
SO$(n+1)/$SO($n-1)$ we find Sasaki-Einstein structures, whereas Sp($n+1)/$Sp($n)=S^{4n+3}$ even carries a 3-Sasakian
structure.
The execptional examples we treated are the nearly K\"ahler space $G_2/SU(3)$ and the nearly parallel $G_2$-manifolds
Spin(7)$/G_2$ and SO(5)/SO(3). It seems likely that every spin manifold of these geometric types can serve as a
supersymmetric heterotic string background, but I leave this to future work.

\bigskip
  There are thus two types of solutions involving the cosets, one on the cone with all fields trivial, and
the one with a linear dilaton. A similar situation exists for type II strings and M-theory, where solutions on
AdS$_p\times X$ ($X$ having Ricci-flat cone) give the near-horizon geometry of branes localized at the
tip of the cone over $X$ \cite{Acharya98}. In our case only the Callan-Harvey-Strominger model on SU(2) appears to be a
brane-limit \cite{Callan91}, whereas it has been argued that the general linear dilaton solutions we considered
can be obtained as a decoupling limit of string theory on the cone, describing the dynamics of states at the
singularity \cite{Aharony98, Giv99}.  

\section*{Acknowledgement} 
 I would like to thank O. Lechtenfeld, A.D. Popov and J.M. Figueroa-O'Farrill for useful comments and discussions,
and the School of Mathematics at the university of Edinburgh for hospitality, where this work was completed. I
acknowledge support from the DFG-Graduiertenkolleg 1463. The work was done within the framework of the project supported
by the DFG under the grant 436 RUS 113/995.

\end{document}